\documentclass[showpacs,amsmath,amssymb,twocolumn,pra,superscriptaddress]{revtex4-2}

\usepackage{amsmath}
\usepackage{amsthm} 
\usepackage{mathtools}
\usepackage{bm}
\usepackage{graphicx}
\usepackage{esvect}
\usepackage{fancyhdr}
\usepackage{graphicx}
\usepackage{graphics}
\usepackage{amssymb}
\usepackage{enumitem}
\usepackage{braket}
\usepackage{xcolor}
\usepackage[colorlinks = true]{hyperref}
\usepackage{subfigure}
\usepackage{tikz}
\usetikzlibrary{decorations.pathreplacing} 

\definecolor{egg}{rgb}{.98,.97,.92}
\definecolor{astroorange}{rgb}{1,.93,.79}
\definecolor{darkorange}{rgb}{1,.89,.6}
\definecolor{dullblue}{rgb}{.29,.47,.77}
\definecolor{grayblue}{rgb}{.98,.98,.98}
\definecolor{fadedblue}{rgb}{.78,.86,.92}
\definecolor{tiffanyblue}{rgb}{.96,1,1}
\definecolor{grayish}{rgb}{.93,.93,.97}
\definecolor{charcoal}{rgb}{.247,.259,.27}
\definecolor{evergreen}{rgb}{.7725,.858,.7647}
\definecolor{dullred}{rgb}{.929,.498,.598}
\definecolor{lavender}{rgb}{.8,.741,.85}

\newcommand{\Hr}{\mathrm{Haar}}
\newcommand{\tr}[1]{\mathrm{Tr}\left\{#1\right\}}
\newcommand{\ptr}[2]{\mathrm{Tr}_{#1}\left\{#2\right\}}

\newcommand{\av}[1]{\underset{\tiny{#1}}{\mathbb{E}}}
\newcommand{\norm}[1]{\left|\left| #1 \right|\right|}
\newcommand{\abs}[1]{\left| #1 \right|}

\newcommand{\OTOC}{\mathrm{OTOC}}

\newcommand{\comment}[1]{}

\newcommand{\xpos}{0} 
\newcommand{\ypos}{0} 
\newcommand{\xrel}{1} 
\newcommand{\yrel}{.5} 
\newcommand{\height}{2em} 
\newcommand{\width}{2em} 
\newcommand{\name}{} 
\newcommand{\nodenum}{0} 
\newcommand{\heightsingle}{2em} 
\newcommand{\heightdouble}{4.6em} 
\newcommand{\widthsingle}{2em} 
\newcommand{\rowspace}{2*\yrel} 

\newtheorem{theorem}{Theorem}
\newtheorem{lemma}{Lemma}
\newtheorem{prop}{Proposition}
\newtheorem{corollary}{Corollary}
\newtheorem{definition}{Definition}

\begin{document}
\title{Resource Theory of Non-Revivals with Applications to Quantum Many-Body Scars}
\author{Roy J. Garcia}
\email{roygarcia@g.harvard.edu}
\affiliation{Department of Physics, Harvard University, Cambridge, Massachusetts 02138, USA}

\author{Kaifeng Bu}
\email{kfbu@fas.harvard.edu}
\affiliation{Department of Physics, Harvard University, Cambridge, Massachusetts 02138, USA}

\author{Liyuan Chen}
\email{liyuanchen@fas.harvard.edu}
\affiliation{Department of Physics, Harvard University, Cambridge, Massachusetts 02138, USA}
\affiliation{John A. Paulson School of Engineering and Applied Science, Harvard University, Cambridge, Massachusetts 02138, USA}

\author{Anton M. Graf}
\email{agraf@g.harvard.edu}
\affiliation{John A. Paulson School of Engineering and Applied Science, Harvard University, Cambridge, Massachusetts 02138, USA}

\date{\today}

\begin{abstract}
The study of state revivals has a long history in dynamical systems. We introduce a resource theory to understand the use of state revivals in quantum physics, especially in quantum many-body scarred systems. In this theory, a state is said to contain no amount of resource if it experiences perfect revivals under some unitary evolution. All other states  are said to be resourceful. We show that this resource bounds information scrambling. Furthermore, we show that quantum many-body scarred dynamics can produce revivals in the Hayden-Preskill decoding protocol and can also be used to recover damaged quantum information. Our theory establishes a framework to study information retrieval and its applications in quantum many-body physics.

\end{abstract}

\maketitle


The study of revivals in quantum dynamics is of fundamental interest in quantum physics. Early studies of state revivals date back to the Poincaré recurrence theorem of 1890, in which a system eventually evolves arbitrarily close to its initial state. In 1957, Bocchieri and Loinger established the quantum recurrence theorem, showing that a system with discrete energies will eventually approach its initial state~\cite{PhysRev.107.337, PhysRevA.18.2379}. Recently, such systems have attracted much attention due to the discovery of long-lived, periodic state revivals in non-integrable systems~\cite{Serbyn2021Quantum, papic2021weak, Moudgalya_2022}. These revivals were first observed experimentally on a Rydberg atom chain~ \cite{Bernien2017Probing}. A further investigation into the PXP model by Turner et al. found that these revivals were due to a set of eigenstates with nearly equal energy spacing~\cite{Turner2018Weak}, dubbed quantum many-body scar (QMBS) states. These eigenstates are non-thermal and exhibit sub-volume law entanglement entropy~\cite{PhysRevB.98.155134}. They constitute a weak violation of the Eigenstate Thermalization Hypothesis~\cite{PhysRevA.43.2046, PhysRevE.50.888}, which posits that eigenstates behave as thermal states.

The discovery of these states has launched an extensive investigation into QMBS systems. A complimentary study explains the revivals via unstable periodic orbits in a phase space~\cite{PhysRevLett.122.040603}, reminiscent of the quantum scars proposed by Heller~\cite{PhysRevLett.53.1515}, which are defined as eigenstates with enhanced probability densities along a classical unstable periodic orbit. This was followed by further investigations into the relation between QMBSs and Heller's scars.~\cite{evrard2023quantum, PhysRevLett.130.250402} By perturbing the PXP model, one can generate perfect revivals \cite{PhysRevLett.122.220603} or reach integrability~ \cite{PhysRevB.99.161101}. Furthermore, an exact form of some QMBS states in the PXP model has been constructed~\cite{PhysRevLett.122.173401}. Analytic expressions for QMBS states were also found in the spin-1 Affleck-Kennedy-Lieb-Tasaki model~\cite{PhysRevB.98.235155}. Other systems exhibiting QMBSs include: the 2-dimensional PXP model~\cite{PhysRevResearch.2.022065, PhysRevB.101.220304}, the transverse field Ising ladder~\cite{PhysRevB.101.220305}, the spin-1 XY model~\cite{PhysRevLett.123.147201, PhysRevB.101.174308}, and a lattice model based on the toric code~\cite{PhysRevResearch.1.033144}.

State revivals in QMBS systems hold potential applications in quantum information, as they can be realized on a quantum simulator~\cite{Bernien2017Probing}. The recovery of initial quantum states is a central problem in quantum information. This task forms the basis of the Hayden-Preskill thought experiment~\cite{Hayden_2007, yoshida2017efficient, Bao2021, yoshida2022recovery}, in which a quantum state is retrieved from a black hole. This experiment has been realized on an ion trap quantum computer~\cite{Landsman2019}. Furthermore, Yan and Sinitsyn put forth a proposal to recover damaged quantum information via chaotic evolution~\cite{PhysRevLett.125.040605}. Both protocols rely on implementing some form of time reversal. This suggests that QMBS systems may hold advantages in information retrieval problems, as they avoid time reversal. Furthermore, QMBS systems were numerically shown to account for revivals of the out-of-time-ordered correlator, indicating that they can recover delocalized information~\cite{PhysRevResearch.4.023095}. It is therefore of interest to develop a resource theory to characterize the quantum revivals found in QMBS systems.

Resource theory is a framework used to quantify advantages provided by a quantum process, called a resource~\cite{RevModPhys.91.025001}. One first defines a set of states which contain the resource, known as resourceful states, and a set of states which contain no resource, known as free states. One then defines a resource monotone, a function which measures the amount of resource in a state. This framework has been used to define resources such as entanglement~\cite{RevModPhys.81.865, RevModPhys.91.025001}, magic~\cite{Veitch_2014, PhysRevLett.118.090501, Wang_2019,BGJ23a,BGJ23b,BGJ23c}, quantum thermodynamics~\cite{PhysRevLett.111.250404, Brandao_secondlaws2015, PhysRevLett.121.110403, PhysRevResearch.2.043374, PhysRevX.11.021014, Chiribella2022, Khanian2023}, coherence~\cite{aberg2006quantifying, PhysRevLett.113.140401, PhysRevLett.116.120404, PhysRevLett.119.150405, RevModPhys.89.041003, PhysRevX.7.011024, tajima2022universalLimit, tajima2022universalTrade},  uncomplexity~\cite{PhysRevA.106.062417}, quantum heat engines~\cite{Bera2021, PhysRevResearch.4.013157}, and scrambling ~\cite{Garcia2022Resource}. This framework has been successful in identifying problems where these resources are useful~\cite{Takag2019,Liu2019}. Entanglement, for example, is used to perform quantum teleportation~\cite{PhysRevLett.70.1895}, while magic is used to bound the time complexity in the classical simulation of quantum circuits~\cite{PhysRevLett.118.090501, PhysRevX.6.021043, Bravyi2019simulationofquantum, PRXQuantum.2.010345, Seddon_2019, BukohPRL19, Bu2022}, and also the
generalization capacity in quantum machine learning~\cite{BuPRA19_stat,BuQST23}.

We utilize this framework to measure a resource which we dub `non-revivals'. Free states are those which experience perfect revivals at a certain time, while all other states are said to be resourceful. We measure the amount of resource in a state by quantifying its lack of revivals. We show that revivals can account for the reversal of scrambling, i.e., information delocalization. We present two applications, showing that QMBS systems: 1) can produce revivals in the Hayden-Preskill decoding protocol and 2) can also be used to recover damaged information.

\textit{Main results.}
We develop a framework to study systems in which some states evolve back to themselves under time evolution. It is known that Hamiltonians with rational eigenvalues generate such revivals. We consider an $n$-qubit Hamiltonian $H$ of dimension $d=2^n$, where $N_R$ ($N_I$) eigenvalues are rational (irrational) numbers. We further assume that $N_R, N_I>1$ and that the rational eigenvalues have a least-common denominator of $T$. We will later show that some quantum many-body scarred models may be rescaled to Hamiltonians with such constraints. We define $S_{\mathrm{Rat}}$ ($S_{\mathrm{Irr}}$) as the set of rational (irrational) eigenstates of $H$, i.e., the set of eigenstates with rational (irrational) energies. We also define $L_{\mathrm{Rat}}=\mathrm{span}(S_{\mathrm{Rat}})$ and the two following sets of indices: $\mathcal{A}=\{i:\ket{\psi_i}\in S_{\mathrm{Rat}}\}$ and $\mathcal{B}=\{i:\ket{\psi_i}\in S_{\mathrm{Irr}}\}$, where $i$ is an index enumerating an eigenstate of $H$.

We assume that the difference between any two irrational eigenergies is irrational. This ensures that only certain quantum states evolve back to themselves, allowing us to define the free states. This property can be found in physical systems. For example, the toy-model Hamiltonian   ${H=\sqrt{2}(X^{\otimes 2}+Z^{\otimes 2})+ Y^{\otimes 2} + \underset{\substack{{P,P'\in \{X,Y,Z\}},\\P\neq P'}}{\sum}P\otimes P'}$, where $X, Y,Z$ are Pauli operators, has both rational and irrational eigenvalues. In the appendix, we show numerically that the difference between any two irrational eigenvalues is irrational, consistent with our assumption.

We now identify the states which evolve back to themselves after a fixed time, which we choose to be $T$. We use the revival fidelity to measure the overlap between an initial state $\ket{\psi}$ and the time evolved state $e^{-iH2\pi T}\ket{\psi}$:
\begin{equation}
	F_R(\ket{\psi})=\abs{\bra{\psi}e^{-iH2\pi T}\ket{\psi}}.
\end{equation}
We say $\ket{\psi}$ exhibits a perfect revival when ${F_R(\ket{\psi})=1}$. This occurs if and only if $e^{-iH2\pi T}\ket{\psi}=\ket{\psi}$ (up to a phase). We introduce a resource theory in which free pure states are those which exhibit perfect revivals. The `resource' is the property which inhibits perfect revivals.

\begin{definition}
Free pure states $\ket{\psi}$ are those satisfying $F_R(\ket{\psi})=1$. Resourceful pure states are all other states.
\end{definition}

We now explicitly construct the set of free pure states. Consider an arbitrary state in $ L_{\mathrm{Rat}}$ with the form ${\ket{\psi}=\sum_{i\in \mathcal{A}}c_i\ket{\psi_i}}$. This state is a superposition of rational eigenstates and exhibits a revival after a time $2\pi T$: ${e^{-iH 2\pi T}\ket{\psi}=\underset{i\in \mathcal{A}}{\sum}e^{-i E_{i}2\pi T}c_i\ket{\psi_i}=\underset{i\in \mathcal{A}}{\sum}c_i\ket{\psi_i}=\ket{\psi}}$. We used that $e^{-iE_i2\pi T}=1$ since $E_iT$ is an integer. In the appendix, we prove that states exhibiting perfect revivals are either eigenstates or have the form of $\ket{\psi}$.

\begin{theorem}\label{Thm:Fidelity}
	The set of free pure states is ${\mathcal{S}_F= S_{\mathrm{Irr}}\sqcup L_{\mathrm{Rat}}}$.
\end{theorem}
We use $\sqcup$ to represent the disjoint union of sets and to stress that $\mathcal{S}_F$ does \textit{not} contain superpositions between states in $S_{\mathrm{Irr}}$ and $L_{\mathrm{Rat}}$. We emphasize that the set of free pure states depends on the Hamiltonian. For example, rescaling a Hamiltonian by a constant may change this set; we discuss this in the appendix. From the above theorem, it follows that a resourceful pure state is a superposition of an irrational eigenstate and at least one other eigenstate.
\begin{corollary}
A resourceful pure state has the form
\begin{equation}
    \ket{\psi_R}=c_j\ket{\psi_j}+\sum_{i\in \mathcal{A}\cup \mathcal{B}}c_{i}\ket{\psi_i},
\end{equation}
where $j\in \mathcal{B}$, $c_j\neq 0$ and there exists at least one $i\neq j$ such that $c_i\neq 0$.
\end{corollary}
To understand why $\ket{\psi_R}$ is resourceful, we evolve this state for a time $2\pi T$, ${e^{-iH2\pi T}\ket{\psi_R}=c_je^{-iE_j2\pi T}\ket{\psi_j}+\underset{i\in \mathcal{A}\cup \mathcal{B}}{\sum}c_{i}e^{-iE_i2\pi T}\ket{\psi_i}}$. The energy $E_j$ is irrational and, by assumption, the difference between $E_j$ and any other eigenenergy must be irrational. Hence, $e^{-iE_j 2\pi T}\neq e^{-iE_i 2\pi T}$ for all $i\neq j$, meaning that ${e^{-iH2\pi T}\ket{\psi_R}\neq \ket{\psi_R}}$. Since $\ket{\psi_R}$ does not exhibit perfect revivals, this state is resourceful.

We now identify the unitaries which can be used to generate the non-revivals resource. We first identify the unitaries which do not increase the amount of resource in a state. These so-called `free' unitaries are defined as those which map free pure states to free pure states; they are denoted by the set $\mathcal{U}_F$. Resourceful unitaries are all other unitaries in the $n$-qubit unitary group $\mathcal{U}$. The evolution unitary $e^{-iHt}$ for a time $t$ is one example of a free unitary. Free unitaries satisfy the following property.
\begin{lemma}\label{Lemma:Free}
If $U\in \mathcal{U}_F$, then $U^\dagger\in \mathcal{U}_F$.  
\end{lemma}
Using Lemma~\ref{Lemma:Free}, we find the form of the free unitaries.
\begin{theorem}\label{Thm:FreeUnitaries}
Free unitaries in $\mathcal{U}_F$ have the form
\begin{equation}
\begin{split}
U_F&=\sum_{\substack{i,j\in \mathcal{A}}}c_{i,j}\ket{\psi_i}\bra{\psi_j}+\sum_{\substack{i\in \mathcal{B}}}c_{i,{\sigma(i)}}\ket{\psi_i}\bra{\psi_{\sigma(i)}},
\end{split}
\end{equation}
where $\sigma$ is a permutation of $\mathcal{B}$, $\abs{c_{i,\sigma(i)}}^2=1$, and $\sum_{i\in \mathcal{A}}c^*_{i,j'}c_{i,j}=\delta_{j',j}$.
\end{theorem}

The first term ensures that states in $L_{\mathrm{Rat}}$ are mapped to $L_{\mathrm{Rat}}$. The second term ensures that irrational eigenstates are mapped to irrational eigenstates. This guarantees that states in $\mathcal{S}_{F}$ are mapped to $\mathcal{S}_{F}$. 

The set of free unitaries is useful in defining functions, known as resource monotones, which measure the amount of resource in a state. We use the revival fidelity to introduce one such function below.
\begin{definition}\label{Def:NonRevivalMonotone}
The non-revival monotone of a pure state $\ket{\psi}$ is
\begin{equation}
	R(\ket{\psi})=\max_{U_F\in \mathcal{U}_F}\left\{
	 1-F_R(U_F\ket{\psi})\right\}.
\end{equation}
\end{definition}
The non-revival monotone satisfies the following two properties, qualifying it as a resource monotone:
\begin{enumerate}
	\item (Faithfulness) $R(\ket{\psi})\geq 0$ for any pure state $\ket{\psi}$ and $R(\ket{\psi})=0$ iff $\ket{\psi}$ is a free pure state. 
	\item (Invariance) $R(U_F\ket{\psi})=R(\ket{\psi})$ where $U_F\in\mathcal{U}_F$ and $\ket{\psi}$ is any pure state.
\end{enumerate}
Faithfulness guarantees that resourceful states generate a positive value of the measure. Invariance guarantees that free unitaries cannot increase the value of the measure, as, by definition, they cannot generate the non-revivals resource. The non-revival monotone has operational meaning; it can bound the experimentally measured revivals of a state $F_R(\ket{\psi})$ through $R(\ket{\psi})\geq 1-F_R(\ket{\psi})$.

It is also useful to quantify the amount of resource generated by a unitary. Resourceful unitaries can map a free pure state to a resourceful pure state. That is, they can destroy the revivals exhibited by some free pure states. This motivates the following operationally meaningful monotone to measure the resource in a unitary.

\begin{definition}
The revival destruction capacity of a unitary $U$ is 
\begin{equation}
	D(U)=\max_{\ket{\psi}\in \mathcal{S}_F}\left\{
R(U\ket{\psi})\right\}.
\end{equation}
\end{definition}
This function is a resource monotone as it satisfies the following conditions: 
\begin{enumerate}
	\item (Faithfulness) $D(U)\geq 0 \ \forall \ U\in \mathcal{\mathcal{U}}$ and $D(U)= 0$ if and only if $U\in \mathcal{U}_F$.
	\item (Invariance) $D(V_1UV_2)=D(U)$ for $U\in\mathcal{U}$ and $V_1,V_2\in \mathcal{U}_F$.
\end{enumerate}
Faithfulness guarantees that the monotone is positive for resourceful unitaries, while invariance ensures that the monotone cannot increase by applying a free unitary.

We have so far developed this resource theory for pure states, as these states are used to study revivals in quantum many-body scarred systems~\cite{Bernien2017Probing}. However, one can also generalize this framework to account for mixed states, which are useful for modeling noise in experiments. The revival fidelity of a density matrix $\rho$ is 
\begin{equation}
	F_{R,M}(\rho)=\tr{\sqrt{\sqrt{\rho}e^{-iH2\pi T}\rho e^{iH2\pi T}\sqrt{\rho}}}.
\end{equation}

\begin{definition} A free density matrix $\rho$ is one which satisfies $F_{R,M}(\rho)=1$. Resourceful density matrices are all other density matrices. 
\end{definition}
We use $\mathcal{M}_F$ to denote the set of free density matrices. We construct an explicit expression for states in this set. 

\begin{theorem}\label{Thm:DensityFidelity}
A free density matrix $\rho\in\mathcal{M}_F$ is a state with the form
\begin{equation}
	\rho=\sum_{\substack{i,j \in  \mathcal{A}\\}}a_{i,j}\ket{\psi_i}\bra{\psi_j}+\sum_{i\in\mathcal{ B}}a_{i,i}\ket{\psi_i}\bra{\psi_i}.
\end{equation}
\end{theorem}
One can check that, with the form of $\rho$ above, $e^{-iH2\pi T}\rho e^{iH2\pi T}=\rho$, yielding $F_{R,M}(\rho)=1$.
Using this theorem, we can find the corresponding free unitaries for density matrices.
\begin{theorem}\label{Thm:FreeUnitariesDensity}
$\mathcal{U}_F$ is the set of unitaries which maps $\mathcal{M}_F$ to itself.
\end{theorem}
By using $\mathcal{U}_F$ and $F_{R,M}$, one can generalize the non-revival monotone in Definition~\ref{Def:NonRevivalMonotone} to density matrices.

\begin{figure*}[t]
\subfigure[]{
\scalebox{.6}{
\begin{tikzpicture}
    \draw [thick,color=charcoal]
    (0,\yrel+2*\rowspace)--(3.5*\xrel,\yrel+2*\rowspace)
    (0,\yrel+\rowspace)--(3.5*\xrel,\yrel+\rowspace)
    (0,\yrel)--(2.5*\xrel,\yrel)
    (2.5*\xrel,1.1*\yrel)--(3.5*\xrel,1.1*\yrel)
    (2.5*\xrel,0.9*\yrel)--(3.5*\xrel,0.9*\yrel)
    (0,-\yrel)--(2.5*\xrel,-\yrel)
    (2.5*\xrel,-1.1*\yrel)--(3.5*\xrel,-1.1*\yrel)
    (2.5*\xrel,-0.9*\yrel)--(3.5*\xrel,-0.9*\yrel)
    (0,-\yrel-\rowspace)--(3.5*\xrel,-\yrel-\rowspace)
    (0,-\yrel-2*\rowspace)--(3.5*\xrel,-\yrel-2*\rowspace);
    
	\renewcommand{\xpos}{-.2*\xrel}
	\draw [decorate,line width=.75pt,color=charcoal,decoration={brace,amplitude=5pt},xshift=\xrel,yshift=-\rowspace]	(\xpos,-\yrel+2*\rowspace) -- (\xpos,+2*\rowspace+\yrel) node [black,midway,xshift=9pt] {};
	    
	\renewcommand{\xpos}{-.2*\xrel}
	\draw [decorate,line width=.75pt,color=charcoal,decoration={brace,amplitude=5pt},xshift=\xrel,yshift=-\rowspace]	(\xpos,-\yrel-0*\rowspace) -- (\xpos,-0*\rowspace+\yrel) node [black,midway,xshift=9pt] {};
	
	\renewcommand{\xpos}{-.2*\xrel}
	\draw [decorate,line width=.75pt,color=charcoal,decoration={brace,amplitude=5pt},xshift=\xrel,yshift=-\rowspace]	(\xpos,-\yrel-2*\rowspace) -- (\xpos,-2*\rowspace+\yrel) node [black,midway,xshift=9pt] {};
    
    \renewcommand{\nodenum}{v15}
    \renewcommand{\name}{\includegraphics[width=.025\textwidth]{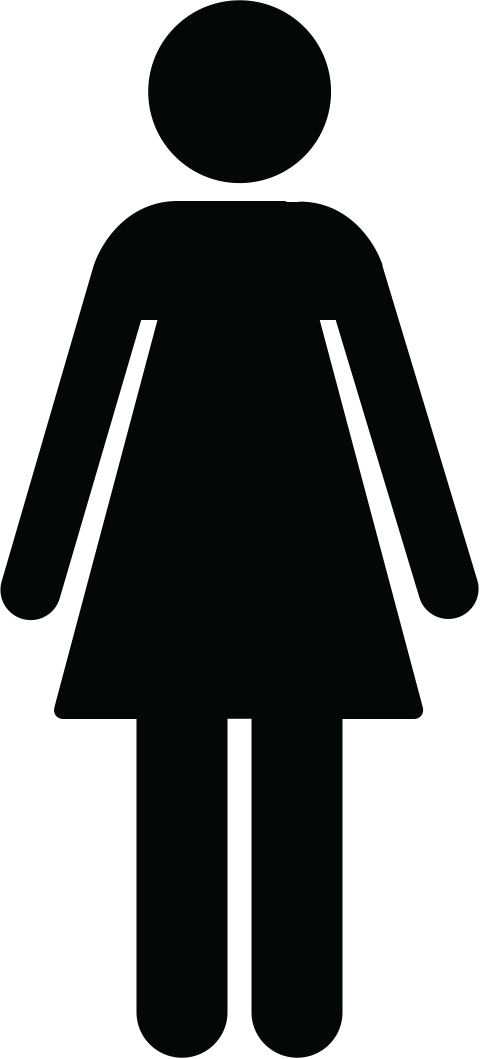}}
	\renewcommand{\xpos}{-2.5*\xrel}
    \renewcommand{\ypos}{\yrel+\rowspace}
    \renewcommand{\height}{\heightdouble}
    \renewcommand{\width}{\widthsingle}
    \node[] (\nodenum) at (\xpos,\ypos) {\name};
     
    \renewcommand{\nodenum}{v15}
    \renewcommand{\name}{\includegraphics[width=.04\textwidth]{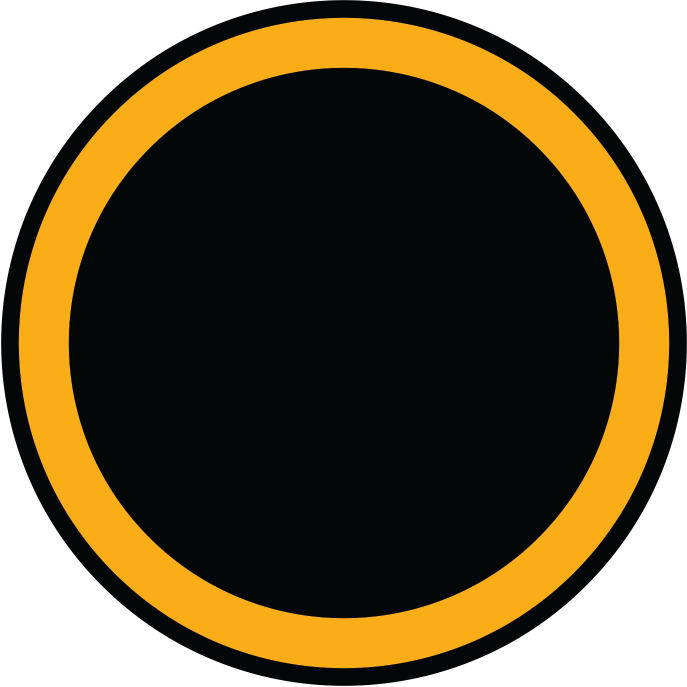}}
	\renewcommand{\xpos}{-2.5*\xrel}
    \renewcommand{\ypos}{\yrel}
    \renewcommand{\height}{\heightdouble}
    \renewcommand{\width}{\widthsingle}
    \node[] (\nodenum) at (\xpos,\ypos) {\name};
        
    \renewcommand{\nodenum}{v15}
    \renewcommand{\name}{\includegraphics[width=.02\textwidth]{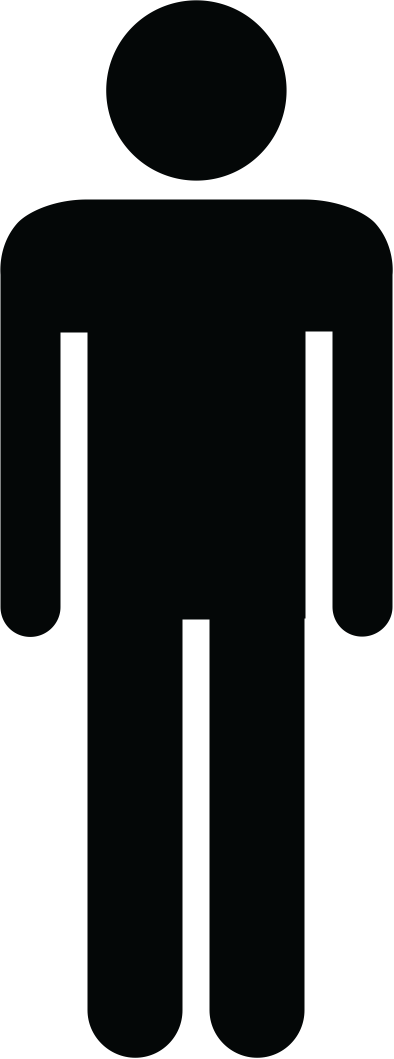}}
	\renewcommand{\xpos}{-2.5*\xrel}
    \renewcommand{\ypos}{-\yrel-\rowspace}
    \renewcommand{\height}{\heightdouble}
    \renewcommand{\width}{\widthsingle}
    \node[] (\nodenum) at (\xpos,\ypos) {\name};
       	
  	\renewcommand{\nodenum}{v2}
    \renewcommand{\name}{$R$}
	\renewcommand{\xpos}{-1.6*\xrel}
    \renewcommand{\ypos}{5*\yrel}
    \renewcommand{\height}{\heightsingle}
    \renewcommand{\width}{\widthsingle}
    \node[rectangle,, line width=.35mm, fill=evergreen, rounded corners, minimum width=\width, minimum height=\height, draw=charcoal] (\nodenum) at (\xpos,\ypos) {\name}; 

  	\renewcommand{\nodenum}{v2}
    \renewcommand{\name}{$A$}
	\renewcommand{\xpos}{-1.6*\xrel}
    \renewcommand{\ypos}{3*\yrel}
    \renewcommand{\height}{\heightsingle}
    \renewcommand{\width}{\widthsingle}
    \node[rectangle,, line width=.35mm, fill=evergreen, rounded corners, minimum width=\width, minimum height=\height, draw=charcoal] (\nodenum) at (\xpos,\ypos) {\name}; 
    
	\renewcommand{\nodenum}{v3}
    \renewcommand{\name}{$B$}
	\renewcommand{\xpos}{-1.6*\xrel}
    \renewcommand{\ypos}{\yrel}
    \renewcommand{\height}{\heightsingle}
    \renewcommand{\width}{\widthsingle}
    \node[rectangle,, line width=.35mm, fill=evergreen, rounded corners, minimum width=\width, minimum height=\height, draw=charcoal] (\nodenum) at (\xpos,\ypos) {\name}; 

	\renewcommand{\nodenum}{v3}
    \renewcommand{\name}{$B'$}
	\renewcommand{\xpos}{-1.6*\xrel}
    \renewcommand{\ypos}{\yrel-\rowspace}
    \renewcommand{\height}{\heightsingle}
    \renewcommand{\width}{\widthsingle}
    \node[rectangle,, line width=.35mm, fill=evergreen, rounded corners, minimum width=\width, minimum height=\height, draw=charcoal] (\nodenum) at (\xpos,\ypos) {\name}; 
    
	\renewcommand{\nodenum}{v3}
    \renewcommand{\name}{$A'$}
	\renewcommand{\xpos}{-1.6*\xrel}
    \renewcommand{\ypos}{\yrel-2*\rowspace}
    \renewcommand{\height}{\heightsingle}
    \renewcommand{\width}{\widthsingle}
    \node[rectangle,, line width=.35mm, fill=evergreen, rounded corners, minimum width=\width, minimum height=\height, draw=charcoal] (\nodenum) at (\xpos,\ypos) {\name}; 
      
	\renewcommand{\nodenum}{v3}
    \renewcommand{\name}{$R'$}
	\renewcommand{\xpos}{-1.6*\xrel}
    \renewcommand{\ypos}{\yrel-3*\rowspace}
    \renewcommand{\height}{\heightsingle}
    \renewcommand{\width}{\widthsingle}
    \node[rectangle,, line width=.35mm, fill=evergreen, rounded corners, minimum width=\width, minimum height=\height, draw=charcoal] (\nodenum) at (\xpos,\ypos) {\name}; 

	\renewcommand{\nodenum}{v4}
    \renewcommand{\name}{ $\ket{\mathrm{Bell}}$}
	\renewcommand{\xpos}{-.75*\xrel}
    \renewcommand{\ypos}{4*\yrel}
    \renewcommand{\height}{\heightsingle}
    \renewcommand{\width}{\widthsingle}
    \node[](\nodenum) at (\xpos,\ypos) {\name};
        
	\renewcommand{\nodenum}{v4}
    \renewcommand{\name}{ $\ket{\mathrm{Bell}}$}
	\renewcommand{\xpos}{-.75*\xrel}
    \renewcommand{\ypos}{0*\yrel}
    \renewcommand{\height}{\heightsingle}
    \renewcommand{\width}{\widthsingle}
    \node[](\nodenum) at (\xpos,\ypos) {\name};
    
	\renewcommand{\nodenum}{v5}
    \renewcommand{\name}{$\ket{\mathrm{Bell}}$}
	\renewcommand{\xpos}{-.75*\xrel}
    \renewcommand{\ypos}{-2*\rowspace}
    \renewcommand{\height}{9.7em}
    \renewcommand{\width}{\widthsingle}
    \node[](\nodenum) at (\xpos,\ypos) {\name};  
	
	\renewcommand{\nodenum}{v6}
    \renewcommand{\name}{$U(t)$}
	\renewcommand{\xpos}{\xrel}
    \renewcommand{\ypos}{\rowspace}
    \renewcommand{\height}{\heightdouble}
    \renewcommand{\width}{\widthsingle}
    \node[rectangle, fill=egg,  line width =.3mm,rounded corners, minimum width=\width, minimum height=  	    \height, draw=charcoal] (\nodenum) at (\xpos,\ypos) {\name};
    
    \renewcommand{\nodenum}{v7}
    \renewcommand{\name}{$U^*(t)$}
	\renewcommand{\xpos}{\xrel}
    \renewcommand{\ypos}{-\rowspace}
    \renewcommand{\height}{\heightdouble}
    \renewcommand{\width}{\widthsingle}
    \node[rectangle, fill=egg,  line width =.3mm,rounded corners, minimum width=\width, minimum height=  	    \height, draw=charcoal] (\nodenum) at (\xpos,\ypos) {\name};
    
    \renewcommand{\nodenum}{v15}
    \renewcommand{\name}{\includegraphics[width=.025\textwidth]{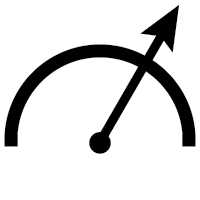}}
	\renewcommand{\xpos}{2.5*\xrel}
    \renewcommand{\ypos}{0}
    \renewcommand{\height}{\heightdouble}
    \renewcommand{\width}{\widthsingle}
    \node[rectangle, fill=white,  line width =.3mm,rounded corners, minimum width=\width, minimum height=  	    \height, draw=charcoal] (\nodenum) at (\xpos,\ypos) {\name};

	\renewcommand{\nodenum}{v17}
    \renewcommand{\name}{$R$}
	\renewcommand{\xpos}{4*\xrel}
    \renewcommand{\ypos}{\yrel+2*\rowspace}
    \renewcommand{\height}{\heightsingle}
    \renewcommand{\width}{\widthsingle}
    \node[rectangle,, line width=.35mm, fill=evergreen, rounded corners, minimum width=\width, minimum height=\height, draw=charcoal] (\nodenum) at (\xpos,\ypos) {\name};
    
	\renewcommand{\nodenum}{v17}
    \renewcommand{\name}{$C$}
	\renewcommand{\xpos}{4*\xrel}
    \renewcommand{\ypos}{\yrel+\rowspace}
    \renewcommand{\height}{\heightsingle}
    \renewcommand{\width}{\widthsingle}
    \node[rectangle,, line width=.35mm, fill=evergreen, rounded corners, minimum width=\width, minimum height=\height, draw=charcoal] (\nodenum) at (\xpos,\ypos) {\name};
    
	\renewcommand{\nodenum}{v17}
    \renewcommand{\name}{$D$}
	\renewcommand{\xpos}{4*\xrel}
    \renewcommand{\ypos}{\yrel}
    \renewcommand{\height}{\heightsingle}
    \renewcommand{\width}{\widthsingle}
    \node[rectangle,, line width=.35mm, fill=evergreen, rounded corners, minimum width=\width, minimum height=\height, draw=charcoal] (\nodenum) at (\xpos,\ypos) {\name};
        
	\renewcommand{\nodenum}{v17}
    \renewcommand{\name}{$D'$}
	\renewcommand{\xpos}{4*\xrel}
    \renewcommand{\ypos}{\yrel-\rowspace}
    \renewcommand{\height}{\heightsingle}
    \renewcommand{\width}{\widthsingle}
    \node[rectangle,, line width=.35mm, fill=evergreen, rounded corners, minimum width=\width, minimum height=\height, draw=charcoal] (\nodenum) at (\xpos,\ypos) {\name};
        
	\renewcommand{\nodenum}{v17}
    \renewcommand{\name}{$C'$}
	\renewcommand{\xpos}{4*\xrel}
    \renewcommand{\ypos}{\yrel-2*\rowspace}
    \renewcommand{\height}{\heightsingle}
    \renewcommand{\width}{\widthsingle}
    \node[rectangle,, line width=.35mm, fill=evergreen, rounded corners, minimum width=\width, minimum height=\height, draw=charcoal] (\nodenum) at (\xpos,\ypos) {\name};
            
	\renewcommand{\nodenum}{v17}
    \renewcommand{\name}{$R'$}
	\renewcommand{\xpos}{4*\xrel}
    \renewcommand{\ypos}{\yrel-3*\rowspace}
    \renewcommand{\height}{\heightsingle}
    \renewcommand{\width}{\widthsingle}
    \node[rectangle,, line width=.35mm, fill=evergreen, rounded corners, minimum width=\width, minimum height=\height, draw=charcoal] (\nodenum) at (\xpos,\ypos) {\name};

\end{tikzpicture}
}}
\subfigure[]{
\includegraphics[width=.21\textwidth]{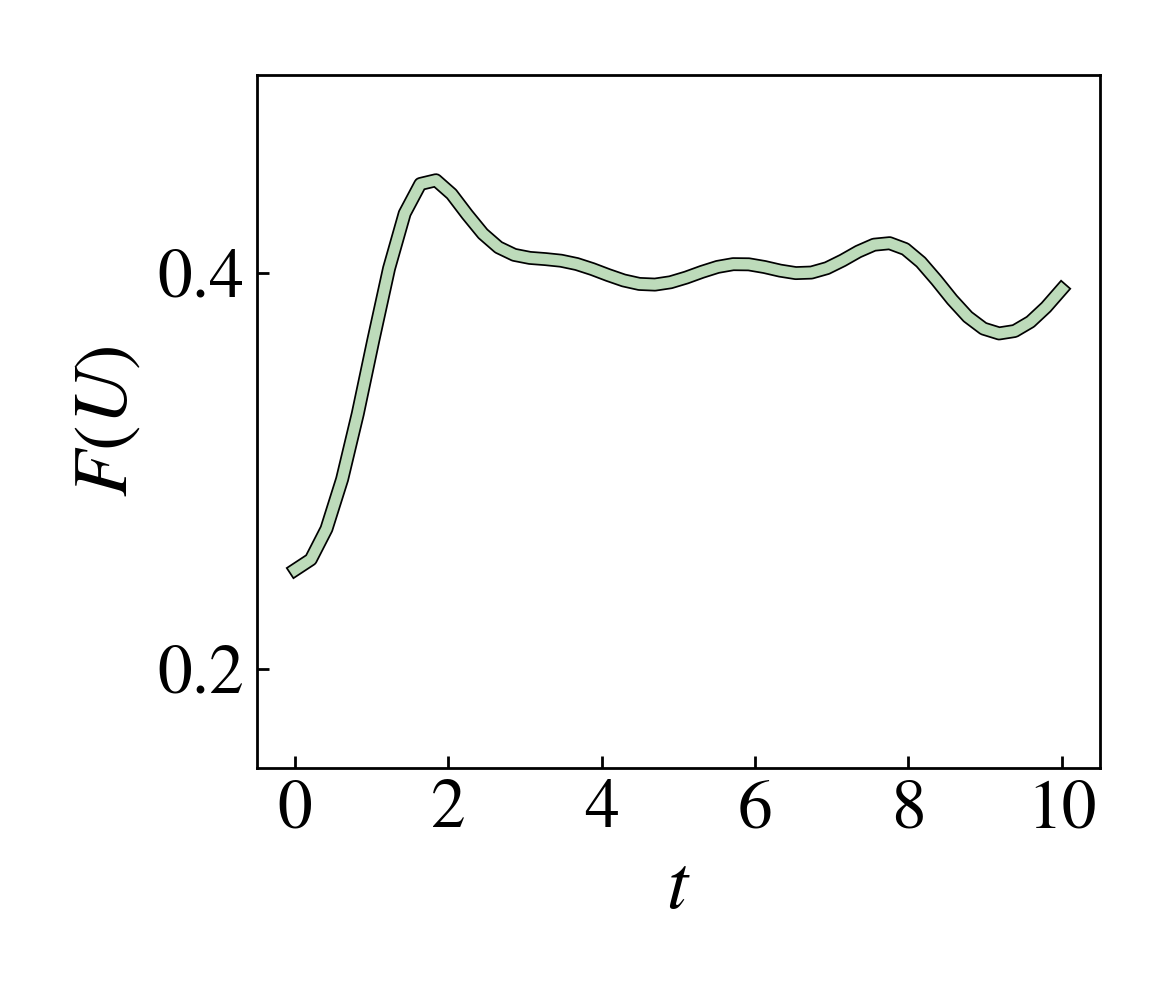}
}
\subfigure[]{
\includegraphics[width=.21\textwidth]{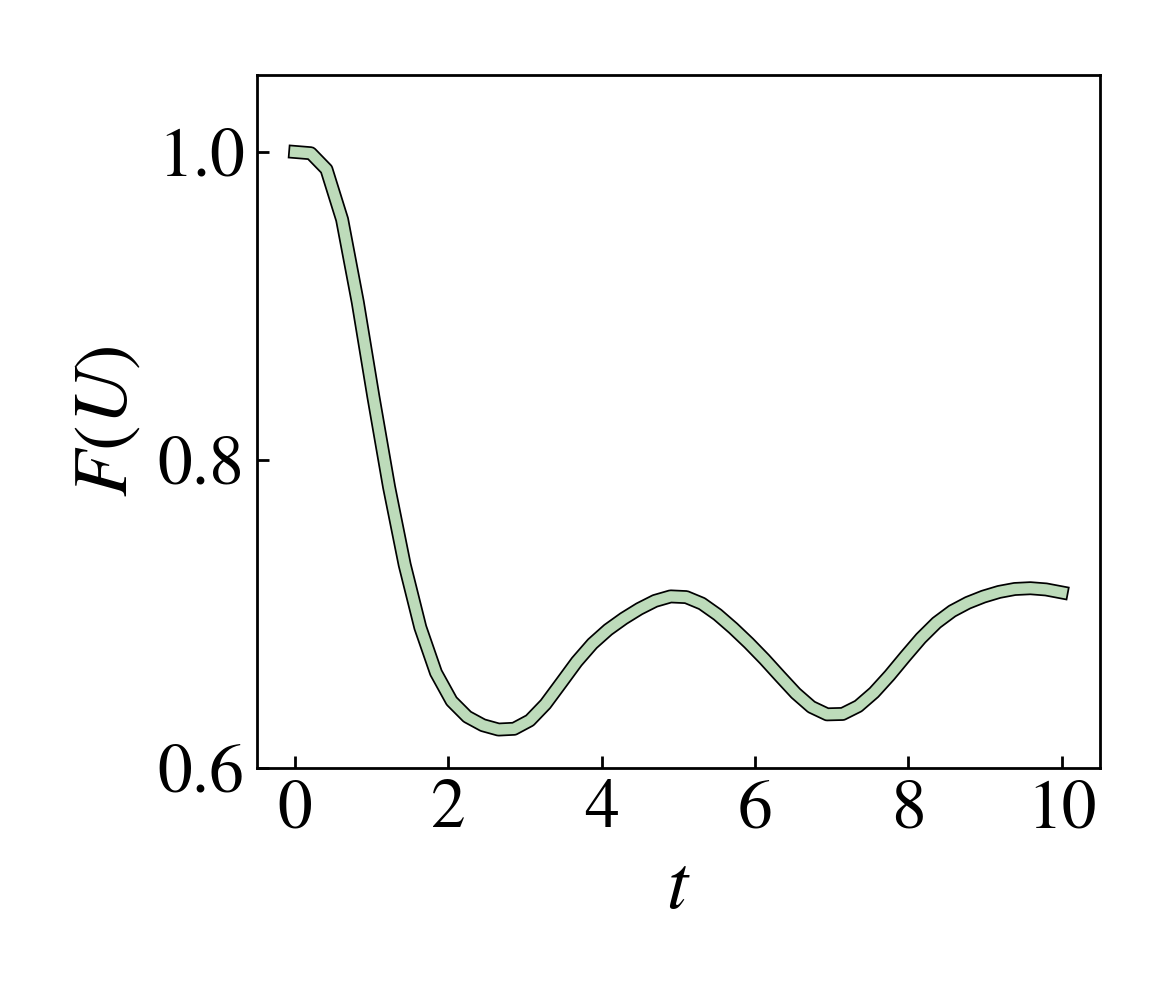}
}
\subfigure[]{
\includegraphics[width=.21\textwidth]{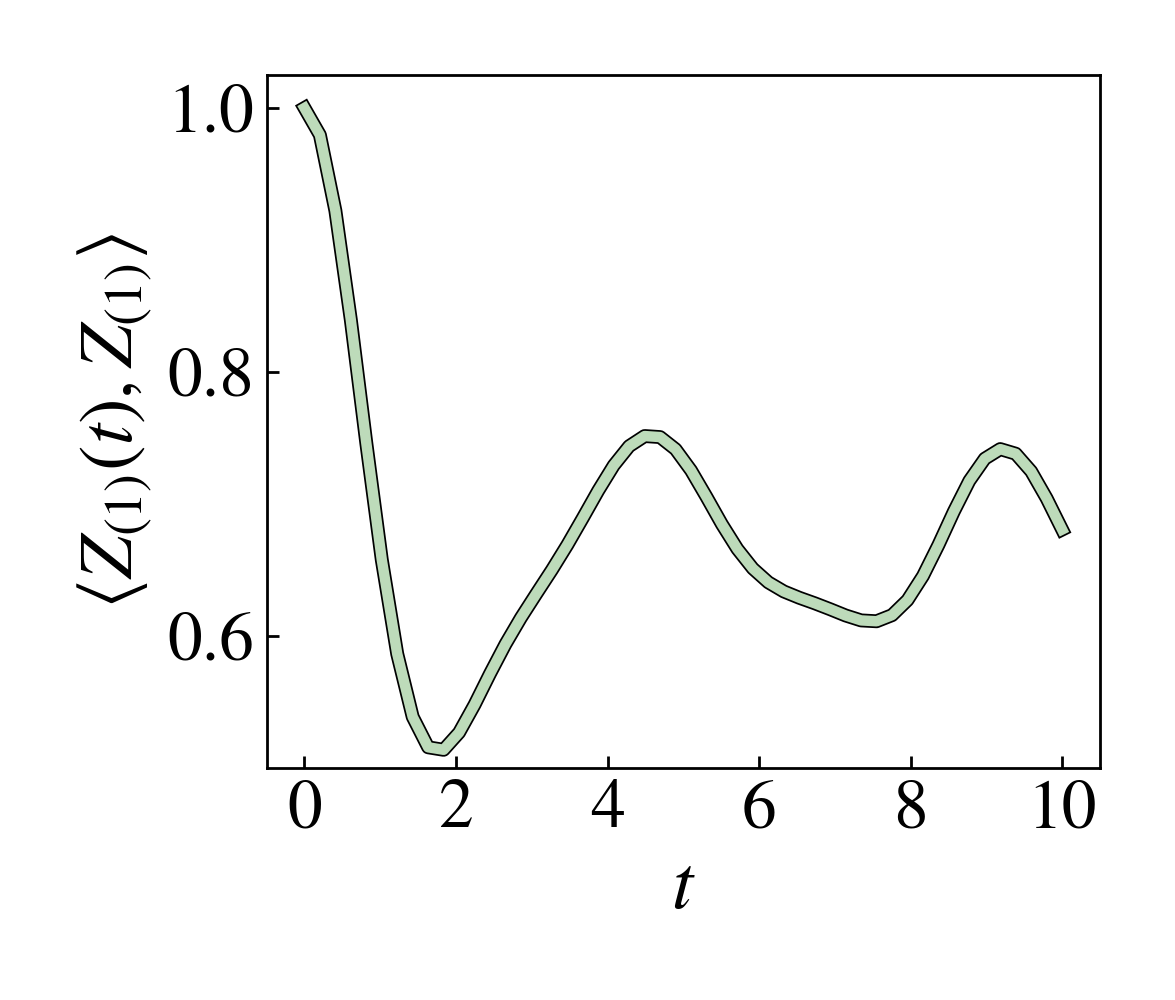}
}
\vspace{-3mm}
\caption{(a) Hayden-Preskill decoding protocol. Alice creates a Bell state between her system $A$ and a reference system $R$. She throws her half of the Bell state into the black hole, system $B$, which forms a Bell state with another system $B'$, representing radiation. The system $AB$ then evolves with $U(t)$. Bob creates a Bell state between his system $A'$ and a reference system $R'$. He then evolves the system $B'A'$ with $U^*(t)$. Bob then projectively measures a Bell state between the $d_D$-dimensional output systems $D$ and $D'$, ${\ket{\mathrm{Bell}}_{D,D'}=\frac{1}{\sqrt{d_D}}\sum_{i=1}^{d_D}\ket{i}_D\otimes \ket{i}_{D'}}$. Afterwards, Bob's reference system $R'$ forms a Bell state with $R$, indicating a successful teleportation of Alice's state. (b) Decoding fidelity of the Hayden-Preskill protocol as a function of time under evolution of the PXP model, where system $A$ is the first qubit and $D$ is the $n$-th qubit. (c) Decoding fidelity where both $D$ and $A$ are the first qubit. (d) Plot of $\langle  Z_{(1)} (t), Z_{(1)}\rangle$ as a function of time. In all plots, the system size is $n=12$.}
\vspace{-3mm}
\label{Figure:PXP_numerics}
\end{figure*}
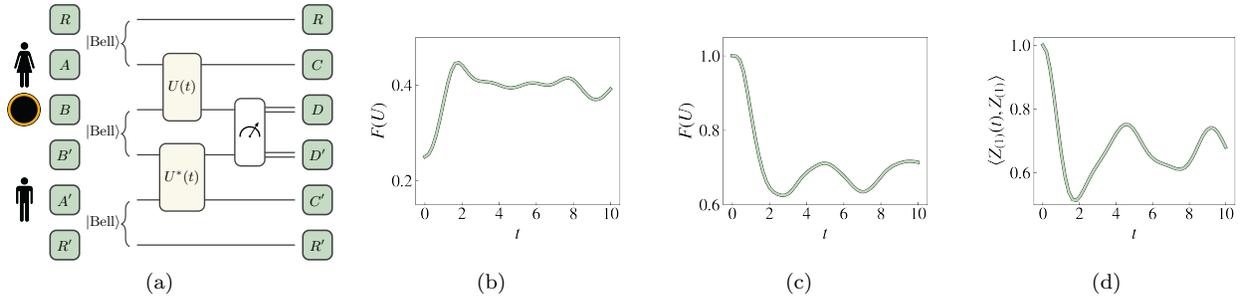

Aside from states, observables can also experience revivals when they undergo Heisenberg evolution. We study the non-revivals resource in observables and its connection to quantum chaos. Define the inner product between the $n$-qubit operators $O_1$ and $O_2$ as $\langle O_1,O_2 \rangle=\frac{1}{d}\tr{O_1^\dagger O_2}$. Define the norm $\norm{\cdot}_2=\sqrt{\frac{1}{d}\tr{|\cdot|^2}}$. We use the notation ${\langle \cdot\rangle =\frac{1}{d}\tr{\cdot}}$ and ${O(t)=e^{iHt}O e^{-iHt}}$. We denote the $n$-qubit Pauli group as $\mathcal{P}_2^{\otimes n}$. Define $\mathcal{O}$ as the set of $n$-qubit observables $O$ satisfying $\norm{O}_2=1$.  Define the revival correlator of an operator $O$  as
\begin{equation}
	G(O)=\abs{\langle O(2\pi T),O\rangle}^2.
\end{equation}
For $O\in \mathcal{O}$, $G(O)=1$ if and only if $O=e^{i\phi}O(2\pi T)$ for some phase $\phi$. Free observables are those which exhibit perfect revivals, up to a phase.
\begin{definition}
The set of free observables, $\mathcal{O}_F$, consists of observables $O\in \mathcal{O}$ satisfying $G(O)=1$. Resourceful observables are all other normalized observables.
\end{definition}
Similar to the case of free density matrices, one can find the expression for the set of free observables.

\begin{prop}\label{Prop:FreeObservable}
A free observable $O_F\in \mathcal{O}_F$ is a normalized observable with the form
\begin{equation}
	O_F=\sum_{i,j\in \mathcal{A}}a_{i,j}\ket{\psi_i}\bra{\psi_j}+\sum_{i\in \mathcal{B}}a_{i,i}\ket{\psi_i}\bra{\psi_i}.
\end{equation}
\end{prop}
Similar to Theorem~\ref{Thm:FreeUnitariesDensity}, $\mathcal{U}_F$ is the set of unitaries which maps all free observables to free observables.

We define the following monotone to quantify the amount of resource in a normalized observable.

\begin{definition}
The observable non-revival monotone of $O\in\mathcal{O}$ is
\begin{equation}
	\mathcal{G}(O)=\max_{U_F\in \mathcal{U}_F}\left\{
1-G(U_F^\dagger O U_F)\right\}.
\end{equation}
\end{definition}
This function satisfies
\begin{enumerate}
	\item (Faithfulness) $\mathcal{G}(O)\geq 0 \ \forall \ O\in \mathcal{O}$ and $\mathcal{G}(O)= 0$ if and only if $O\in \mathcal{O}_F$.
	\item (Invariance) $\mathcal{G}(U_F^\dagger O U_F)=\mathcal{G}(O)$ $\forall \ U_F\in \mathcal{U}_F$.
\end{enumerate}
This function is positive for resourceful observables (due to faithfulness) and does not increase if we evolve an observable by a free unitary (due to invariance).

We show how this monotone relates to scrambling, i.e., information spreading. The out-of-time-ordered correlator (OTOC) is a function used to experimentally measure scrambling, which marks the onset of chaos in quantum many-body systems. We bound the OTOC with the observable non-revival monotone, giving our monotone operational meaning. For a unitary $U$, define the OTOC between two operators $O_1$ and $O_2$ as
\begin{equation}
	\OTOC(O_1,O_2;U)=\langle U^\dagger O_1UO_2U^\dagger O_1UO_2\rangle.
\end{equation}
When $O_1$ and $O_2$ are disjoint, an OTOC value near zero is used as an indication that information from $O_1$ has spread to the support of $O_2$, quantifying information scrambling. For $U(t)=e^{-iH t}$ and either $O_1\in \mathcal{O}_F$ or $O_2\in \mathcal{O}_F$, it is clear that ${\OTOC(O_1,O_2;U(2\pi T))=\OTOC(O_1,O_2;U(0))}$. We take the case where $O_1$ and $O_2$ are Pauli operators. If $U(t)$ is scrambling, $O_1(t)$ becomes a uniform sum of many Pauli operators, yielding an OTOC value near 0. Conversely, if $ O_1(t)$ retains a large overlap with $O_1$, i.e., if a revival occurs, then one expects the OTOC to be close to its initial value and scrambling is suppressed. This motivates the following bound.

\begin{theorem}\label{Theorem:OTOCBound}
	For a Pauli operator $O_1$,
\begin{equation}
\begin{split}
	1-2\mathcal{G}(O_1) \leq \OTOC(O_1,O_1;U(2\pi T))
	&\leq  1.
\end{split}
\end{equation}	
\end{theorem}

When $O_1$ exhibits substantial revivals, $\mathcal{G}(O_1)$ is small, which can lead to larger values of $\OTOC(O_1,O_1;U(2\pi T))$ near its initial value of 1. These OTOC revivals, indicative of `unscrambled' information, are consistent with previous numerical studies of scrambling in QMBS systems~\cite{PhysRevResearch.4.023095}.

We present applications of our resource theory to two information recovery problems. In both applications, we consider a QMBS system with the Hamiltonian $H_{\mathrm{QMBS}}$. Assume that all eigenvalues of $H_{\mathrm{QMBS}}$ are irrational. In this case, the scar states of $H_{\mathrm{QMBS}}$ are a set of eigenstates with equally-spaced eigenvalues of the form $\{E_0m_i+E_1\}_{i=1}^{N_R}$, where $m_i$ is a non-zero integer and either $E_0$ or $E_1$ (or both) is an irrational number. 

We rescale $H_{\mathrm{QMBS}}$ to the new Hamiltonian, ${H'_{\mathrm{QMBS}}=(H_{\mathrm{QMBS}}-E_1)/E_0}$, so that it meets the constraints required by our resource theory, i.e., multiple rational and irrational eigenvalues. The Hamiltonian $H'_{\mathrm{QMBS}}$ has $N_R$ rational eigenvalues, $\{m_i\}_{i=1}^{N_R}$, and a set of $N_R$ rational eigenstates, $S_{\mathrm{Rat}}$. We assume that the remaining eigenvalues are irrational and that the difference between any two irrational eigenvalues is irrational. We define the evolution unitary as $U(t)=e^{-iH'_{\mathrm{QMBS}}t}$.

In our first application, we use our resource theory to bound the success of the Hayden-Preskill decoding protocol (explained in Fig.~\ref{Figure:PXP_numerics} (a)), in which a quantum state is thrown into and recovered from a  black hole. We assume that the black hole dynamics are given by the QMBS unitary, $U(t)$. This assumption is motivated by the fact that both this protocol~\cite{Landsman2019} and QMBS dynamics~ \cite{Bernien2017Probing} have been demonstrated on quantum devices. The success of the recovery is given by the decoding fidelity 
\begin{equation}\label{Eq:OTOCFidelity}
	F(U(t))=\frac{1}{d_A^2 \av{P_A, P_D}\OTOC(P_A,P_D;U(t))},
\end{equation}
where $P_A$ ($P_D$) is a Pauli operator on system $A$ ($D$), and $d_A$ is the dimension of system $A$. The average is taken uniformly over all Pauli operators on their respective systems. At $t=0$, $F(U(0))=\frac{1}{d_{A\backslash D}^2 }$, where $d_{A\backslash D}=2^{\abs{A\backslash D}}$. When $A$ and $D$ are disjoint, $F(U(0))=\frac{1}{d_A^2}$.

We numerically study the dynamics of $F(U(t))$. We consider the $n$-qubit PXP model with periodic boundary conditions: $H_{\mathrm{PXP}}=\sum_{i=1}^n  \Pi_{(i-1)} X_{(i)} \Pi_{(i+1)}$,
where $\Pi^0_{(i)}=\ket{0}\bra{0}_i$, and $X_{(i)}$ and $Z_{(i)}$ are Pauli operators on the $i$-th qubit. The evolution unitary is $U(t)=e^{-iH_{\mathrm{PXP}}t}$. This QMBS system contains scar states with approximately equal energy spacing. We do not rescale this Hamiltonian, as this will only rescale the time parameter in the numerical simulations.

We show that PXP dynamics produce oscillations in $F(U(t))$, varying the success of recovery. We consider two scenarios, finding that: (1) when $A$ and $D$ are disjoint, $F(U(t))$ attains local minima when the unitary unscrambles information (see Fig.~\ref{Figure:PXP_numerics} (b)); (2) when $A$ and $D$ are the same, $F(U(t))$ attains local maxima when the unitary unscrambles information (see Fig.~\ref{Figure:PXP_numerics} (c)). To explain these findings, we plot $\langle Z_{(1)}(t),Z_{(1)}\rangle$ in Fig.~\ref{Figure:PXP_numerics} (d), showing that the PXP dynamics produce revivals of the Pauli operator $Z_{(1)}$. When these revivals occur, $\av{P_A, P_D}\OTOC(P_A,P_D;U(t))$ (and hence $F(U(t))$) evolves toward its initial value. This implies that when $\langle Z_{(1)}(t),Z_{(1)}\rangle$ attains a peak, Fig.~\ref{Figure:PXP_numerics} (b) attains a trough while Fig.~\ref{Figure:PXP_numerics} (c) attains a peak.

In our second application, we present a cryptography protocol which utilizes QMBS dynamics to recover damaged information. A complementary protocol which uses time-reversal for information recovery is presented in~\cite{PhysRevLett.125.040605}. Alice wishes to send a message $\ket{\phi}$ to Bob, who will use this state to compute $\bra{\phi}O\ket{\phi}$, where $O$ is a single-qubit, resourceful observable. Eve plans to eavesdrop by measuring $\ket{\phi}$ as it is transmitted to Bob. To thwart Eve's efforts and securely transmit the message, Alice and Bob implement the following encoding/decoding scheme.

(1) Alice prepares many copies of the density matrix $\rho_0=\ket{\phi}\bra{\phi}$ where ${\ket{\phi}=\sum_{i\in \mathcal{A}}c_i\ket{\psi_i}}$ is a free pure state.

(2) To hide her information, Alice applies the encoding unitary $U(t_1)=e^{-iH'_{\mathrm{QMBS}} t_1}$, obtaining  ${\rho_1=U(t_1)\rho_0 U^\dagger(t_1)}$, where $t_1$ is selected such that $\abs{\bra{\phi}\rho_1\ket{\phi}}$ is small. She then sends all copies of her state to Bob via a public communication channel.

(3) Eve intercepts the message in the communication channel, performing a weak measurement on $\rho_1$ with a strength of $p\in(0,1)$, which Bob knows. This produces
\begin{equation}
\begin{split}
  \mathcal{M}_p(\rho_1)=&(1-p)\rho_1\otimes \ket{0}\bra{0}_E\\
  &+p\sum_{i=1}^{d}\ket{i}\bra{i}_A\rho_1 \ket{i}\bra{i}_A\otimes \ket{i}\bra{i}_E,
\end{split}
\end{equation}
where $\ket{i}_A$ denotes a computational basis state on Alice's $d$-dimensional system and $\ket{i}_E$ denotes a basis state on Eve's $d+1$-dimensional system. Intuitively, a weak measurement can be interpreted as a probabilistic projective measurement; more details can be found in the appendix.

(4) After receiving the state $ \mathcal{M}_p(\rho_1)$, Bob applies the decoding unitary $U(t_2)$, where $t_2=mt_R-t_1>0$, $m$ is an integer, and $t_R$ is the earliest time at which ${U(t_R)\ket{\phi}=\ket{\phi}}$. The final reduced state on Alice's original system, where we use $\ket{i}=\ket{i}_A$, is
\begin{equation}\label{Eq:rhof}
\begin{split}
\rho_{f}
	&=(1-p)\ket{\phi}\bra{\phi}  \\
	&\hspace{5mm}+p\sum_{i=1}^{d}U (t_2)\ket{i}\bra{i}U^\dagger (t_2)\abs{\bra{i}U(t_1)\ket{\phi}}^2.
\end{split}
\end{equation}

(5) Bob uses $\rho_f$ to construct $\bra{\phi}O\ket{\phi}$.

To accomplish step (5), we first use  Eq.~\eqref{Eq:rhof} to write
\begin{equation}
\begin{split}
    \bra{\phi}O\ket{\phi} =&\frac{1}{1-p}[\tr{O\rho_{f}} \\
 &-p\sum_{i=1}^{d}\bra{i}U^\dagger (t_2) O U (t_2)\ket{i}\abs{\bra{i}U(t_1)\ket{\phi}}^2].
\end{split}
\end{equation}
We assume that each computational basis state $\ket{i}$ has a small overlap with the scar states (rational eigenstates) of $H_{\mathrm{QMBS}}'$. Thus $U(t_2)$ does not generate revivals of $\ket{i}$. Instead, $U(t_2)$ evolves $\ket{i}$ chaotically, yielding ${\bra{i}U^\dagger (t_2) O U (t_2)\ket{i}\rightarrow \int_{\Hr} dU\bra{i}U^\dagger O U\ket{i}=\frac{1}{d}\tr{O}}$
when $t_2$ is large. The expectation value simplifies to
\begin{equation}\label{Eq:ExpectFinal}
	\bra{\phi}O\ket{\phi}=\frac{1}{1-p}\left(\tr{O\rho_{f}}-\frac{p}{d}\tr{O}\right).
\end{equation}

By experimentally measuring $\tr{O\rho_{f}}$, Bob can obtain $\bra{\phi}O\ket{\phi}$ through Eq.~\eqref{Eq:ExpectFinal}. Thus, Bob can recover some form of Alice's initial information in spite of Eve's interference. Furthermore, the measurement of $\tr{O\rho_{f}}$ can be performed efficiently via techniques such as classical shadow tomography \cite{Huang2020}. In the PXP model, the assumption on $\bra{i}U^\dagger (t_2) O U (t_2)\ket{i}$ does not necessarily hold, as some computational basis states have a large overlap with scarred eigenstates. Thus, it is interesting to explore models which permit this assumption.

\textit{Summary.} We developed a resource theory of quantum non-revivals. States and observables with no amount of resource are those which experience perfect periodic revivals. We measure the amount of resource in a state or observable by quantifying the suppression of revivals. We show that such measures can be used to understand the reversal of information scrambling via the out-of-time-ordered correlator. This suggests that one may also be able to bound magic with our framework, which we leave as an open problem. We also show that the revivals present in quantum many-body scarred systems can both enable and inhibit the retrieval of quantum information. For future consideration, we propose exploring the connection between our framework and other resource theories, such as the theory of coherence.

\textit{Acknowledgements.}
We are grateful to Arthur Jaffe and Eric J. Heller for insightful discussions. This work was supported in part by the ARO Grant W911NF-19-1-0302 and the ARO
MURI Grant W911NF-20-1-0082. A.M.G thanks the Harvard Quantum Initiative for financial support.

\bibliography{Bibliography}

\newpage
\onecolumngrid
\begin{appendix}

\section{Toy model for irrational spacing assumption}
In the main text, we make the assumption that, for a given Hamiltonian, no two irrational eigenvalues have rational spacing. We provide a toy model
\begin{equation}
H=\sqrt{2}(X\otimes X+Z\otimes Z)+ Y\otimes Y + \underset{\substack{{P,P'\in \{X,Y,Z\}},\\P\neq P'}}{\sum}P\otimes P'
\end{equation}
which we state satisfies this condition. We provide more details in this appendix. We numerically compute the eigenvalues of this Hamiltonian: $\{3,-1-\sqrt{2},-1+\sqrt{2}-\sqrt{10},-1+\sqrt{2}+\sqrt{10}\}$. The first eigenvalue in this set is rational, while the remaining eigenvalues are irrational. The difference between any two irrational values is an irrational number. In the main text, we focus on systems which have more than one rational eigenvalue, as these systems have more practical applications for encoding quantum information and are relevant to quantum many-body scarred systems. However, one can extend the results of this work to include systems with only one rational eigenvalue.

\section{Proof of Theorem~\ref{Thm:Fidelity}}
Consider an arbitrary state in $ L_{\mathrm{Rat}}$,
\begin{equation}
	\ket{\psi_F}=\sum_{i\in \mathcal{A}}c_i\ket{\psi_i}\in L_{\mathrm{Rat}}.
\end{equation}
The eigenvalue of a rational state $\ket{\psi_i}$ can be written as $\tfrac{n_i}{T}$ for an integer $n_i$. Therefore,
\begin{equation}
	e^{-iH2\pi T}\ket{\psi_F}=\sum_{i\in \mathcal{A}}c_ie^{-i2\pi n_i}\ket{\psi_i}=\sum_{i\in \mathcal{A}}c_i\ket{\psi_i}=\ket{\psi_F}.
\end{equation}
It follows that $F_R(\ket{\psi}_F)=1$. In the case where we instead choose $\ket{\psi_F}=\ket{\psi_i}\in S_{\mathrm{Irr}}$, then ${e^{-iH2\pi T}\ket{\psi_F}=e^{-iE_i2\pi T}\ket{\psi_F}}$ and $F_R(\ket{\psi}_F)=1$. It follows that all states in the set $S_{\mathrm{Irr}}\sqcup L_{\mathrm{Rat}}$ have a revival fidelity of unity.

Now consider the case where we have a state $\ket{\psi_R}\notin S_{\mathrm{Irr}}\sqcup L_{\mathrm{Rat}}$ with the form
\begin{equation}\label{Eq:psiR}
	\ket{\psi_R}=\sum_{\substack{i\in \mathcal{B}'\subseteq \mathcal{B},\\ \abs{\mathcal{B}'}> 1}}c_i\ket{\psi_i},
\end{equation}
where all coefficients are nonzero. Then
\begin{equation}
	e^{-iH2\pi T}\ket{\psi_R}=\sum_{\substack{i\in \mathcal{B}'\subseteq \mathcal{B},\\ \abs{\mathcal{B}'}>1}}c_ie^{-iE_i 2\pi T}\ket{\psi_i}\neq \ket{\psi_R}.
\end{equation}
This follows since all $E_i$ in this sum are irrational and the difference between any two irrational eigenenergies is assumed to be irrational, implying that $e^{-iE_i 2\pi T}\neq e^{-iE_j 2\pi T}$ for $i\neq j$. Therefore,  $F_R(\ket{\psi_R})<1$. For concreteness, we provide an example of a state with the form of Eq.~\eqref{Eq:psiR}. Take $\ket{\psi_R}=\frac{1}{\sqrt{2}}(\ket{\psi_1}+\ket{\psi_2})$, where $\ket{\psi_1}$ and $\ket{\psi_2}$ are irrational eigenstates. The recovery fidelity for this state is $F_R(\ket{\psi_R}) = \sqrt{\frac{1}{2}(1+\cos((E_1-E_2)2\pi T))}$. Since $E_1-E_2$ is an irrational number, then $\cos((E_1-E_2)2\pi T)\neq 1$, so $F_R(\ket{\psi_R})<1$.

Lastly, we take the case where a state $\ket{\psi_R} \notin S_{\mathrm{Irr}}\sqcup L_{\mathrm{Rat}}$ has the form
\begin{equation}
	\ket{\psi_R}=\sum_{\substack{i\in \mathcal{A}'\subseteq \mathcal{A}, \\\abs{\mathcal{A}'}> 0}}c_i\ket{\psi_i}+\sum_{\substack{i\in \mathcal{B}'\subseteq \mathcal{B},\\ \abs{\mathcal{B}'}>0}}c_i\ket{\psi_i},
\end{equation}
where all coefficients are nonzero. Since all eigenstates in the second sum have irrational energies, then 
\begin{equation}
	e^{-iH2\pi T}\ket{\psi_R}=\sum_{\substack{i\in \mathcal{A}'\subseteq \mathcal{A}, \\\abs{\mathcal{A}'}> 0}}c_i\ket{\psi_i}+\sum_{\substack{i\in \mathcal{B}'\subseteq \mathcal{B},\\ \abs{\mathcal{B}'}>0}}c_i e^{-iE_i 2\pi T}\ket{\psi_i}\neq \ket{\psi_R}
\end{equation}
and $F_R(\ket{\psi_R})<1$. Hence, $F_R(\ket{\psi_R})=1$ if and only if  $\ket{\psi}\in S_{\mathrm{Irr}}\sqcup L_{\mathrm{Rat}}$, yielding $\mathcal{S}_F=S_{\mathrm{Irr}}\sqcup L_{\mathrm{Rat}}$.

\section{Set of Free Pure states of Rescaled Hamiltonian}
In the main text, we state that rescaling a Hamiltonian by a constant can change the set of free pure states. Here, we provide an example. Consider the Hamiltonian, 
${H=\ket{00}\bra{00}+2\ket{01}\bra{01}+\sqrt{2}\ket{10}\bra{10}+2\sqrt{2}\ket{11}\bra{11}}$. The rational eigenstates are $\ket{00}$ and $\ket{01}$, since they have rational eigenenergies of 1 and 2, respectively. The irrational eigenstates are $\ket{10}$ and $\ket{11}$, since they have irrational eigenenergies of $\sqrt{2}$ and $2\sqrt{2}$, respectively. The set of free pure states is therefore $\{\ket{10},\ket{11}\}\sqcup\mathrm{span}(\{\ket{00},\ket{01}\})$. By rescaling $H$ to $\sqrt{2}H$, the rational eigenstates of the new Hamiltonian are $\ket{10}$ and $\ket{11}$, and the set of free pure states is $\{\ket{00},\ket{01}\}\sqcup\mathrm{span}(\{\ket{10},\ket{11}\})$. We use this example to demonstrate that the resource theory is Hamiltonian-dependent.

\section{Proof of Lemma~\ref{Lemma:Free}}
By definition, for a free unitary $U$, $\{U\ket{\psi_1}:\ket{\psi_1}\in \mathcal{S}_F\}\subseteq \mathcal{S}_F$. Unitaries are one-to-one mappings, i.e. for any two states $\ket{\psi_1}$ and $\ket{\psi_2}$ satisfying $\ket{\psi_1}\neq \ket{\psi_2}$, then $U \ket{\psi_1}\neq U \ket{\psi_2}$. Therefore, given a free unitary $U$, $\{U\ket{\psi_1}:\ket{\psi_1}\in \mathcal{S}_F\}=\mathcal{S}_F$. For any fixed free pure state $\ket{\psi_2}$, there exists a free pure state $\ket{\psi_1}$ such that $\ket{\psi_2}=U\ket{\psi_1}$. Hence $U^\dagger \ket{\psi_2}=\ket{\psi_1}\in\mathcal{S}_F$ for all $\ket{\psi_2}\in \mathcal{S}$, implying that $U^\dagger \in \mathcal{U}_F$.

\section{Proof of Theorem~\ref{Thm:FreeUnitaries}}
Let $U_F\in \mathcal{U}_F$. This unitary can be written as
\begin{equation}
	U_F=\sum_{i,j\in \mathcal{A} \cup \mathcal{B}}c_{i,j}\ket{\psi_i}\bra{\psi_j}
\end{equation}
for some coefficients $c_{i,j}$ which we now find.
For a fixed $k\in \mathcal{A}$,
\begin{equation}
\begin{split}
	U_F\ket{\psi_k}
	&=\sum_{i,j\in \mathcal{A}\cup \mathcal{B}}c_{i,j}\ket{\psi_i}\bra{\psi_j}\psi_k\rangle\\
	&=\sum_{i,j\in \mathcal{A}\cup \mathcal{B}}c_{i,j}\ket{\psi_i}\delta_{j,k}\\	
	&=\sum_{i\in \mathcal{A}\cup \mathcal{B}}c_{i,k}\ket{\psi_i}.
\end{split}
\end{equation}
We require that $U_F\ket{\psi_k}\in \mathcal{S}_F=S_{\mathrm{Irr}}\sqcup L_{\mathrm{Rat}}$. We first take the case where $U_F\ket{\psi_k}\in S_{\mathrm{Irr}}$ and show that this leads to a contradiction. In this case, there exists an index $j\in \mathcal{B}$ such that $\abs{c_{j,k}}=1$ and $c_{i,k}=0$ for $i\neq j$ and $i\in \mathcal{A}\cup \mathcal{B}$. Since $N_R>1$, then there is an index $i\in \mathcal{A}$ satisfying $i\neq j$, such that there exists a free pure state $c_1\ket{\psi_i}+c_2\ket{\psi_k}$ which satisfies
\begin{equation}
\begin{split}
U_F\left[c_1\ket{\psi_i}+c_2\ket{\psi_k}\right]
	&=c_1U_F\ket{\psi_i}+c_2\ket{\psi_j}\\
	&\notin \mathcal{S}_F.
\end{split}
\end{equation}
This implies that $U_F\notin \mathcal{U}_F$, which is a contradiction. Therefore, it must be true that $U_F\ket{\psi_k}\in L_{\mathrm{Rat}}$ for all $k\in \mathcal{A}$. Hence, $c_{i,k}=0$ when $i\in \mathcal{B}$ and $k\in \mathcal{A}$. 

Since $U_F^\dagger \in \mathcal{U}_{F}$, then it also holds that $U_F^\dagger \ket{\psi_k}\in  L_{\mathrm{Rat}}$ for all $k\in \mathcal{A}$. Therefore
\begin{equation}
\begin{split}
	U_F^\dagger \ket{\psi_k}
	&=\sum_{i,j\in \mathcal{A}\cup \mathcal{B}}c^*_{i,j}\ket{\psi_j}\bra{\psi_i}\psi_k\rangle\\
	&=\sum_{i,j\in \mathcal{A}\cup \mathcal{B}}c^*_{i,j}\ket{\psi_j}\delta_{i,k}\\	
	&=\sum_{j\in \mathcal{A}\cup \mathcal{B}}c^*_{k,j}\ket{\psi_j}
\end{split}
\end{equation}
implies that $c_{k,j}=0$ when $k\in \mathcal{A}$ and $j\in \mathcal{B}$. The unitary is then 
\begin{equation}
\begin{split}
	U_F
	&=\sum_{\substack{i,j\in \mathcal{A}}}c_{i,j}\ket{\psi_i}\bra{\psi_j}+\sum_{\substack{i\in \mathcal{B},\\ j\in \mathcal{A}}}c_{i,j}\ket{\psi_i}\bra{\psi_j}+\sum_{\substack{i\in \mathcal{A},\\ j\in \mathcal{B}}}c_{i,j}\ket{\psi_i}\bra{\psi_j}+\sum_{\substack{i,j\in \mathcal{B}}}c_{i,j}\ket{\psi_i}\bra{\psi_j}\\
	&=\sum_{\substack{i,j\in \mathcal{A}}}c_{i,j}\ket{\psi_i}\bra{\psi_j}+\sum_{\substack{i,j\in \mathcal{B}}}c_{i,j}\ket{\psi_i}\bra{\psi_j}.
\end{split}
\end{equation}

Now take $k\in B$ and compute
\begin{equation}
\begin{split}
	U_F^\dagger \ket{\psi_k}
	&=\left[\sum_{\substack{i,j\in \mathcal{A}}}c^*_{i,j}\ket{\psi_j}\bra{\psi_i}+\sum_{\substack{i,j\in \mathcal{B}}}c^*_{i,j}\ket{\psi_j}\bra{\psi_i}\right]\ket{\psi_k}\\
	&=\sum_{\substack{i,j\in \mathcal{B}}}c^*_{i,j}\ket{\psi_j}\delta_{i,k}\\
	&=\sum_{\substack{i\in \mathcal{B}}}c^*_{k,j}\ket{\psi_j}.
\end{split}
\end{equation}
In order for $U_F^\dagger\ket{\psi_k}\in \mathcal{S}_F$, it must be true that  $U_F^\dagger\ket{\psi_k}\in S_{\mathrm{Irr}}$. For $i,k\in \mathcal{B}$, there exists an index $l_k\in \mathcal{B}$ such that $c_{k,j}=0$ for $j\neq l_k$ and $\abs{c_{k,l_k}}=1$. The unitary is
\begin{equation}
\begin{split}
	U_F
	&=\sum_{\substack{i,j\in \mathcal{A}}}c_{i,j}\ket{\psi_i}\bra{\psi_j}+\sum_{\substack{i,j\in \mathcal{B}}}c_{i,j}\ket{\psi_i}\bra{\psi_j}\\
	&=\sum_{\substack{i,j\in \mathcal{A}}}c_{i,j}\ket{\psi_i}\bra{\psi_j}+\sum_{\substack{i\in \mathcal{B}}}c_{i,l_i}\ket{\psi_i}\bra{\psi_{l_i}}.
\end{split}
\end{equation}

We find the conditions on $U_F$ which make it unitary:
\begin{equation}
\begin{split}
	U_F^\dagger U_F
	&=\left[\sum_{i',j'\in \mathcal{A}}c^*_{i',j'}\ket{\psi_{j'}}\bra{\psi_{i'}}+\sum_{i'\in \mathcal{B}}c^*_{i',l_{i'}}\ket{\psi_{l_{i'}}}\bra{\psi_{i'}}\right]\left[\sum_{i,j\in \mathcal{A}}c_{i,j}\ket{\psi_i}\bra{\psi_j}+\sum_{i\in \mathcal{B}}c_{i,l_i}\ket{\psi_i}\bra{\psi_{l_i}}\right]\\
	&=\sum_{i,j,i',j'\in \mathcal{A}}c^*_{i',j'}c_{i,j}\ket{\psi_{j'}}\bra{\psi_{i'}}\psi_i\rangle \bra{\psi_j}+\sum_{\substack{i,i'\in \mathcal{B}}}c^*_{i',l_{i'}}c_{i,l_i}\ket{\psi_{l_{i'}}}\bra{\psi_{i'}}\psi_{i}\rangle \bra{\psi_{l_i}}\\
	&=\sum_{i,j,i',j'\in \mathcal{A}}c^*_{i',j'}c_{i,j}\delta_{i',i}\ket{\psi_{j'}} \bra{\psi_j}+\sum_{\substack{i,i'\in \mathcal{B}}}c^*_{i',l_{i'}}c_{i,l_i}\delta_{i',i}\ket{\psi_{l_{i'}}} \bra{\psi_{l_i}}\\
	&=\sum_{i,j,j'\in \mathcal{A}}c^*_{i,j'}c_{i,j}\ket{\psi_{j'}} \bra{\psi_j}+\sum_{\substack{i\in \mathcal{B}}}c^*_{i,l_i}c_{i,l_i}\ket{\psi_{l_i}} \bra{\psi_{l_i}}\\
	&=\sum_{j,j'\in \mathcal{A}}\left[\sum_{i\in \mathcal{A}}c^*_{i,j'}c_{i,j}\right]\ket{\psi_{j'}} \bra{\psi_j}+\sum_{\substack{i\in \mathcal{B}}}\abs{c_{i,l_i}}^2\ket{\psi_{l_i}} \bra{\psi_{l_i}}.
\end{split}
\end{equation}
Unitarity requires that $\{l_i\}_{i\in B}=\mathcal{B}$, $\abs{c_{i,l_i}}^2=1$ for $i\in \mathcal{B}$, and $\sum_{i\in \mathcal{A}}c^*_{i,j'}c_{i,j}=\delta_{j',j}$. Furthermore, the condition $\{l_i\}_{i\in B}=\mathcal{B}$ implies that $l_i$ can be written as $l_i=\sigma(i)$, where $\sigma$ is a permutation of $\mathcal{B}$.

\section{Properties of non-revival monotone}
Since the revival fidelity satisfies $F_R(\ket{\psi})\leq 1$ for any $\ket{\psi}$, then $R(\ket{\psi})\geq 0$. Take the case where $\ket{\psi}$ is a free pure state. Then for any $U_F\in \mathcal{U}_F$, $U_F\ket{\psi}$ is also a free pure state and $F_R(U_F\ket{\psi})=1$, implying that $R(\ket{\psi})=0$. Conversely, if $R(\ket{\psi})=0$, then $F_R(U_F\ket{\psi})=1$ for all $U_F\in \mathcal{U}_F$. This is true only if $\ket{\psi}$ is a free pure state. This proves faithfulness. Invariance follows from the definition of the non-revival monotone.

\section{Properties of revival destruction capacity}
	We have $D(U)\geq 0$ since $R(\ket{\psi})\geq 0$. If $U\in \mathcal{U}_F$, then for $\ket{\psi}\in \mathcal{S}_F$, $U\ket{\psi}\in \mathcal{S}_F$. This produces $R(U\ket{\psi})=0$ and $D(U)=0$. If $U\notin \mathcal{U}_F$, then there exists a state $\ket{\psi}\in \mathcal{S}_F$ such that $U\ket{\psi}\notin \mathcal{S}_F$, yielding $R(U\ket{\psi})>0$ and $D(U)>0$. Therefore, $D(U)=0$ if and only if $U\in \mathcal{U}_F$, proving faithfulness. 
	
We prove invariance. Assuming that  $V_1,V_2\in \mathcal{U}_F$,
\begin{equation}
\begin{split}
D(V_1UV_2)
	&=\max_{\substack{\ket{\psi}\in \mathcal{S}_F,\\ U_F\in \mathcal{U}_F}}\left\{
	 1-F_R(U_FV_1UV_2\ket{\psi})\right\}\\
	&=\max_{\substack{\ket{\psi'}\in \mathcal{S}_F,\\ U_F'\in \mathcal{U}_F}}\left\{
	 1-F_R(U_F'U\ket{\psi'})\right\}\\
	&=D(U).
\end{split}
\end{equation}
In the second line, we use that the maximization is invariant under $V_2\ket{\psi}=\ket{\psi'}\in \mathcal{S}_F$ and $U_FV_1=U_F'\in\mathcal{U}_F$.

\section{Proof of Theorem~\ref{Thm:DensityFidelity}}
Define the state
\begin{equation}
	\rho=\sum_{\substack{i,j \in  \mathcal{A}\\}}a_{i,j}\ket{\psi_i}\bra{\psi_j}+\sum_{i\in\mathcal{ B}}a_{i,i}\ket{\psi_i}\bra{\psi_i}.
\end{equation}
Then
\begin{equation}
\begin{split}
		\rho(2\pi T)
		&=\sum_{\substack{i,j \in  \mathcal{A}\\}}a_{i,j}e^{-i(E_i-E_j)2\pi T}\ket{\psi_i}\bra{\psi_j}+\sum_{i\in\mathcal{ B}}a_{i,i}e^{-i(E_i-E_i)2\pi T}\ket{\psi_i}\bra{\psi_i}\\
		&=\sum_{\substack{i,j \in  \mathcal{A}\\}}a_{i,j}\ket{\psi_i}\bra{\psi_j}+\sum_{i\in\mathcal{ B}}a_{i,i}\ket{\psi_i}\bra{\psi_i}\\
		&=\rho.
\end{split}
\end{equation}
Hence, $F_{R,M}(\rho)=1$. 

We now show that all density matrices  satisfying $F_{R,M}(\rho)=1$ having the above form. Any quantum state can be written as
\begin{equation}
	\rho=\sum_{\substack{i,j \in  \mathcal{A}\\}}a_{i,j}\ket{\psi_i}\bra{\psi_j}+\sum_{i\in \mathcal{B}}a_{i,i}\ket{\psi_i}\bra{\psi_i}+\sum_{i,j\in \mathcal{B}, i<j}(a_{i,j}\ket{\psi_i}\bra{\psi_j}+a_{j,i}\ket{\psi_j}\bra{\psi_i})+\sum_{i\in \mathcal{A},j\in \mathcal{B}}(a_{i,j}\ket{\psi_i}\bra{\psi_j}+a_{j,i}\ket{\psi_j}\bra{\psi_i}).
\end{equation}
We evolve
\begin{equation}
\begin{split}
	\rho(2\pi T)
	&=\sum_{\substack{i,j \in  \mathcal{A}\\}}a_{i,j}\ket{\psi_i}\bra{\psi_j}+\sum_{i\in \mathcal{B}}a_{i,i}\ket{\psi_i}\bra{\psi_i}\\
	&\hspace{10mm}+\sum_{i,j\in \mathcal{B}, i<j}(a_{i,j}e^{-i(E_i-E_j)2\pi T}\ket{\psi_i}\bra{\psi_j}+a_{j,i}e^{i(E_i-E_j)2\pi T}\ket{\psi_j}\bra{\psi_i})\\
	&\hspace{10mm}+\sum_{i\in \mathcal{A},j\in \mathcal{B}}(a_{i,j}e^{-i(E_i-E_j)2\pi T}\ket{\psi_i}\bra{\psi_j}+a_{j,i}e^{i(E_i-E_j)2\pi T}\ket{\psi_j}\bra{\psi_i})\\
	&=\sum_{\substack{i,j \in  \mathcal{A}\\}}a_{i,j}\ket{\psi_i}\bra{\psi_j}+\sum_{i\in \mathcal{B}}a_{i,i}\ket{\psi_i}\bra{\psi_i}\\
	&\hspace{10mm}+\sum_{i,j\in \mathcal{B}, i<j}e^{-i(E_i-E_j)2\pi T}(a_{i,j}\ket{\psi_i}\bra{\psi_j}+a_{j,i}e^{i(E_i-E_j)4\pi T}\ket{\psi_j}\bra{\psi_i})\\
	&\hspace{10mm}+\sum_{i\in \mathcal{A},j\in \mathcal{B}}e^{iE_j2\pi T}(a_{i,j}\ket{\psi_i}\bra{\psi_j}+a_{j,i}e^{-iE_j4\pi T}\ket{\psi_j}\bra{\psi_i}).
\end{split}
\end{equation}
We find the conditions which ensure that $\rho(2\pi T)=\rho$ so that $F(\rho)=1$. In the third and fourth sums, if $a_{i,j}$ is nonzero, then so is $a_{j,i}$ by the Hermiticity condition, $a_{i,j}=a_{j,i}^*$. In the fourth sum, $e^{-iE_j 4\pi T}\neq 1$. Hence, if $a_{j,i}\neq 0$ then $(a_{i,j}\ket{\psi_i}\bra{\psi_j}+a_{j,i}e^{-iE_j4\pi T}\ket{\psi_j}\bra{\psi_i})\neq (a_{i,j}\ket{\psi_i}\bra{\psi_j}+a_{j,i}\ket{\psi_j}\bra{\psi_i})$, implying that $\rho(2\pi T)\neq \rho$. Therefore all coefficients in the fourth sum must vanish. Similarly, since we assume that there is no rational spacing between irrational eigenvalues, then  $e^{i(E_{i}-E_j)4\pi T}\neq 1$ in the third sum. Hence all coefficients in the third sum must vanish. This proves that a state $\rho$ satisfying $F_{R,M}(\rho)=1$ must have the form of
\begin{equation}
	\rho=\sum_{\substack{i,j \in  \mathcal{A}\\}}a_{i,j}\ket{\psi_i}\bra{\psi_j}+\sum_{i\in\mathcal{ B}}a_{i,i}\ket{\psi_i}\bra{\psi_i}.
\end{equation}

\section{Proof of Theorem~\ref{Thm:FreeUnitariesDensity}}

We define $\mathcal{U}_{\mathcal{M}_F}$ as the set of unitaries which map $\mathcal{M}_F$ to itself. We will show that $\mathcal{U}_{\mathcal{M}_F}=\mathcal{U}_F$. We first take a unitary $U\in \mathcal{U}_{\mathcal{M}_F}$ to have the general form
\begin{equation}
U=\sum_{i,j\in \mathcal{A}\cup \mathcal{B}}c_{i,j}\ket{\psi_i}\bra{\psi_j}.
\end{equation}
We define the density matrix $\rho_k=\ket{\psi_k}\bra{\psi_k}\in \mathcal{M}_F$, where $k\in \mathcal{A} \cup \mathcal{B}$. We compute the following
\begin{equation}
\begin{split}
U^\dagger \rho_k U
	&=\left[	\sum_{i',j'\in \mathcal{A}\cup \mathcal{B}}c^*_{i',j'}\ket{\psi_{j'}}\bra{\psi_{i'}}\right]\ket{\psi_k}\bra{\psi_k} \left[\sum_{i,j\in \mathcal{A}\cup\mathcal{B}}c_{i,j}\ket{\psi_i}\bra{\psi_j}\right]\\
	&=\sum_{i',j',i,j\in \mathcal{A}\cup \mathcal{B}}c^*_{i',j'}c_{i,j}\ket{\psi_{j'}}\bra{\psi_j}\delta_{i',k}\delta_{k,i}\\
	&=\sum_{j',j\in \mathcal{A}\cup \mathcal{B}}c^*_{k,j'}c_{k,j}\ket{\psi_{j'}}\bra{\psi_j}.
\end{split}
\end{equation}
Since $U^\dagger \rho_k U\in \mathcal{M}_F$, it follows that $c^*_{k,j'}c_{k,j}=0$ for $k\in \mathcal{A}\cup \mathcal{B}$ and $j,j'\in \mathcal{B}$ where $j\neq j'$. 
This implies that there exists an index $l_k\in \mathcal{B}$ such that $c_{k,j}=0$ for $j\in \mathcal{B}$ and $j\neq l_k$. The subscript in $l_k$ denotes that this index may depend on the value of $k$.

Before proceeding with the proof, we introduce the following lemma.
\begin{lemma}\label{Lemma:FreeUnitaryDM}
If $U\in \mathcal{U}_{\mathcal{M}_F}$, then $U^\dagger \in \mathcal{U}_{\mathcal{M}_F}$.
\end{lemma}
\begin{proof}
By definition, for a unitary $U\in\mathcal{U}_{\mathcal{M}_F}$,  $\{U^\dagger \rho U : \rho\in \mathcal{M}_F\}\subseteq \mathcal{M}_F$. Unitaries are one-to-one mappings, implying that for any two free pure states $\rho_1$ and $\rho_2$, if $\rho_1\neq \rho_2$, then $U^\dagger \rho_1 U \neq U^\dagger \rho_2 U$. Therefore, $\{U^\dagger \rho_1 U : \rho_1\in \mathcal{M}_F\}=\mathcal{M}_F$. For any fixed free observable $\rho_2$, there exists a free pure state $\rho_1$ such that $\rho_2=U^\dagger \rho_1 U$. Hence $U\rho_2U^\dagger=\rho_1\in \mathcal{M}_F$ for all $\rho_2\in\mathcal{M}_F$ and $U^\dagger \in \mathcal{U}_{\mathcal{M}_F}$.
\end{proof}
Since $U^\dagger\in \mathcal{U}_{\mathcal{M}_F}$ by Lemma~\ref{Lemma:FreeUnitaryDM}, then $U \rho_k U^\dagger \in\mathcal{M}_F$. We compute 
\begin{equation}
\begin{split}
U\rho_k U^\dagger
	&=\left[	\sum_{i',j'\in \mathcal{A}\cup \mathcal{B}}c_{i',j'}\ket{\psi_{i'}}\bra{\psi_{j'}}\right]\ket{\psi_k}\bra{\psi_k} \left[\sum_{i,j\in \mathcal{A}\cup \mathcal{B}}c^*_{i,j}\ket{\psi_j}\bra{\psi_i}\right]\\
	&=\sum_{i',j',i,j\in \mathcal{A}\cup \mathcal{B}}c_{i',j'}c^*_{i,j}\ket{\psi_{i'}}\bra{\psi_i}\delta_{j',k}\delta_{k,j}\\
	&=\sum_{i',i\in \mathcal{A}\cup \mathcal{B}}c_{i',k}c^*_{i,k}\ket{\psi_{i'}}\bra{\psi_i}.
\end{split}
\end{equation}
It follows that $c_{i',k}c^*_{i,k}=0$ for $k\in \mathcal{A}\cup \mathcal{B}$ and $i,i'\in \mathcal{B}$ where $i\neq i'$. This implies that there exists an index $m_k\in \mathcal{B}$ such that $c_{i,k}=0$ for $i\in \mathcal{B}$ and $i\neq m_k$.

The unitary becomes
\begin{equation}
	\begin{split}
	U
	&=\sum_{i,j\in \mathcal{A}\cup \mathcal{B}}c_{i,j}\ket{\psi_i}\bra{\psi_j}\\
	&=\sum_{i,j\in \mathcal{A}}c_{i,j}\ket{\psi_i}\bra{\psi_j}+\sum_{i,j\in \mathcal{B}}c_{i,j}\ket{\psi_i}\bra{\psi_j}+\sum_{i\in \mathcal{A},j\in \mathcal{B}}c_{i,j}\ket{\psi_i}\bra{\psi_j}+\sum_{i\in \mathcal{B},j\in \mathcal{A}}c_{i,j}\ket{\psi_i}\bra{\psi_j}\\
	&=\sum_{i,j\in \mathcal{A}}c_{i,j}\ket{\psi_i}\bra{\psi_j}+\sum_{i\in \mathcal{B}}c_{i,l_i}\ket{\psi_i}\bra{\psi_{l_i}}+\sum_{i\in \mathcal{A}}c_{i,l_i}\ket{\psi_i}\bra{\psi_{l_i}}+\sum_{i\in \mathcal{A}}c_{m_i,i}\ket{\psi_{m_i}}\bra{\psi_i},
	\end{split}
\end{equation} 
where $l_i,m_i\in \mathcal{B}$.

We consider the following state
\begin{equation}
\begin{split}
\rho_{k,p}
&=(a_{k}\ket{\psi_k}+a_p\ket{\psi_p})(a_k^*\bra{\psi_k}+a_p^*\bra{\psi_p})\\
&=a_{k,k}\ket{\psi_k}\bra{\psi_k}+a_{k,p}\ket{\psi_k}\bra{\psi_p}+a_{p,k}\ket{\psi_p}\bra{\psi_k}+a_{p,p}\ket{\psi_p}\bra{\psi_p}.
\end{split}
\end{equation}
where $k,p\in \mathcal{A}$ and, for example, $a_{k,p}=a_ka_p^*$. We assume that all coefficients are non-zero. Now compute
\begin{equation}
\begin{split}
U^\dagger \rho_{k,p} U
	&=\left[\sum_{i',j'\in \mathcal{A}}c^*_{i',j'}\ket{\psi_{j'}}\bra{\psi_{i'}}+\sum_{i'\in \mathcal{B}}c^*_{i',l_{i'}}\ket{\psi_{l_{i'}}}\bra{\psi_{i'}}+\sum_{i'\in \mathcal{A}}c^*_{i',l_{i'}}\ket{\psi_{l_{i'}}}\bra{\psi_{i'}}+\sum_{i'\in \mathcal{A}}c^*_{m_{i'},i'}\ket{\psi_{i'}}\bra{\psi_{m_{i'}}}\right]\\
	&\hspace{20mm}\times\left[a_{k,k}\ket{\psi_k}\bra{\psi_k}+a_{k,p}\ket{\psi_k}\bra{\psi_p}+a_{p,k}\ket{\psi_p}\bra{\psi_k}+a_{p,p}\ket{\psi_p}\bra{\psi_p}\right] \\
	&\hspace{20mm}\times\left[\sum_{i,j\in \mathcal{A}}c_{i,j}\ket{\psi_i}\bra{\psi_j}+\sum_{i\in \mathcal{B}}c_{i,l_i}\ket{\psi_i}\bra{\psi_{l_i}}+\sum_{i\in \mathcal{A}}c_{i,l_i}\ket{\psi_i}\bra{\psi_{l_i}}+\sum_{i\in \mathcal{A}}c_{m_i,i}\ket{\psi_{m_i}}\bra{\psi_i}\right]\\
	&=\left[\sum_{i',j'\in \mathcal{A}}c^*_{i',j'}\ket{\psi_{j'}}\bra{\psi_{i'}}+\sum_{i'\in \mathcal{A}}c^*_{i',l_{i'}}\ket{\psi_{l_{i'}}}\bra{\psi_{i'}}\right]
	\left[a_{k,k}\ket{\psi_k}\bra{\psi_k}+a_{k,p}\ket{\psi_k}\bra{\psi_p}+a_{p,k}\ket{\psi_p}\bra{\psi_k}+a_{p,p}\ket{\psi_p}\bra{\psi_p}\right] \\
	&\hspace{20mm}\times\left[\sum_{i,j\in \mathcal{A}}c_{i,j}\ket{\psi_i}\bra{\psi_j}+\sum_{i\in \mathcal{A}}c_{i,l_i}\ket{\psi_i}\bra{\psi_{l_i}}\right]\\
	&=\left[Q_1^\dagger +Q_2^\dagger \right]\rho_{k,p}\left[Q_1 +Q_2 \right]\\
	&=Q_1^\dagger \rho_{k,p}Q_1+Q_1^\dagger \rho_{k,p}Q_2+Q_2^\dagger \rho_{k,p}Q_1+Q_2^\dagger \rho_{k,p}Q_2 ,
\end{split}
\end{equation}
where we define $Q_1=\sum_{i,j\in \mathcal{A}}c_{i,j}\ket{\psi_i}\bra{\psi_j}$ and $Q_2=\sum_{i\in \mathcal{A}}c_{i,l_i}\ket{\psi_i}\bra{\psi_{l_i}}$.
We now prove conditions on the coefficients to ensure that $U^\dagger \rho_{k,p}U\in \mathcal{M}_F$. We compute the fourth term in the above equation:

\begin{equation}
\begin{split}
Q_2^\dagger \rho_{k,p} Q_2
	&=\sum_{i'\in \mathcal{A}}c^*_{i',l_{i'}}\ket{\psi_{l_{i'}}}\bra{\psi_{i'}}
	\left[a_{k,k}\ket{\psi_k}\bra{\psi_k}+a_{k,p}\ket{\psi_k}\bra{\psi_p}+a_{p,k}\ket{\psi_p}\bra{\psi_k}+a_{p,p}\ket{\psi_p}\bra{\psi_p}\right]\sum_{i\in \mathcal{A}}c_{i,l_i}\ket{\psi_i}\bra{\psi_{l_i}}\\
	&=\sum_{i,i'\in \mathcal{A}}c_{i,l_i}c^*_{i',l_{i'}}\ket{\psi_{l_{i'}}}\bra{\psi_{l_i}}
	\left[a_{k,k}\delta_{i',k}\delta_{k,i}+a_{k,p}\delta_{i',k}\delta_{p,i}+a_{p,k}\delta_{i',p}\delta_{k,i}+a_{p,p}\delta_{i',p}\delta_{p,i}\right]\\
	&=a_{k,k}c_{k,l_k}c^*_{k,l_{k}}\ket{\psi_{l_{k}}}\bra{\psi_{l_k}}+a_{k,p}c_{p,l_p}c^*_{k,l_{k}}\ket{\psi_{l_{k}}}\bra{\psi_{l_p}}
	+a_{p,k}c_{k,l_k}c^*_{p,l_{p}}\ket{\psi_{l_{p}}}\bra{\psi_{l_k}}+a_{p,p}c_{p,l_p}c^*_{p,l_{p}}\ket{\psi_{l_{p}}}\bra{\psi_{l_p}}.
\end{split}
\end{equation}

As $l_p,l_k\in\mathcal{B}$, we require that $c_{p,l_p}c^*_{k,l_{k}}=0$ for $l_p\neq l_k$ and for all $p,k\in \mathcal{A}$. This implies that there exists an index $L\in\{l_p\}_{p\in A}$ such that $c_{p,l_p}=0$ for $l_p\neq L$ and $p\in A$. This yields $Q_2=\sum_{i\in \mathcal{A},l_i=L}c_{i,L}\ket{\psi_i}\bra{\psi_{L}}$. We are free to include the terms in the sum where $l_i\neq L$, since $c_{i,l_i}=0$ in this case. We do this for convenience and write $Q_2=\sum_{i\in \mathcal{A}}c_{i,L}\ket{\psi_i}\bra{\psi_{L}}$.

Now we compute the second term:
\begin{equation}
\begin{split}
Q_1^\dagger \rho_{k,p} Q_2
	&=\sum_{i',j'\in \mathcal{A}}c^*_{i',j'}\ket{\psi_{j'}}\bra{\psi_{i'}}
	\left[a_{k,k}\ket{\psi_k}\bra{\psi_k}+a_{k,p}\ket{\psi_k}\bra{\psi_p}+a_{p,k}\ket{\psi_p}\bra{\psi_k}+a_{p,p}\ket{\psi_p}\bra{\psi_p}\right]\sum_{i\in \mathcal{A}}c_{i,L}\ket{\psi_i}\bra{\psi_{L}}\\
	&=\sum_{i,i',j'\in \mathcal{A}}c_{i,L}c^*_{i',j'}\ket{\psi_{j'}}\bra{\psi_{L}}
	\left[a_{k,k}\delta_{i',k}\delta_{k,i}+a_{k,p}\delta_{i',k}\delta_{p,i}+a_{p,k}\delta_{i',p}\delta_{k,i}+a_{p,p}\delta_{i',p}\delta_{p,i}\right]\\
	&=\sum_{j'\in \mathcal{A}}\ket{\psi_{j'}}\bra{\psi_{L}}
	\left[a_{k,k}c_{k,L}c^*_{k,j'}+a_{k,p}c_{p,L}c^*_{k,j'}+a_{p,k}c_{k,L}c^*_{p,j'}+a_{p,p}c_{p,L}c^*_{p,j'}\right]\\
	&=\sum_{j'\in \mathcal{A}}\ket{\psi_{j'}}\bra{\psi_{L}}
	\left[a_{k,k}c_{k,L}c^*_{k,j'}+\abs{a_{k,p}}e^{i
\theta}c_{p,L}c^*_{k,j'}+\abs{a_{k,p}}e^{-i\theta}c_{k,L}c^*_{p,j'}+(1-a_{k,k})c_{p,L}c^*_{p,j'}\right].
\end{split}
\end{equation}
In the last line, we rewrite the coefficient as $a_{k,p}=\abs{a_{k,p}}e^{i\theta}$ and use that $a_{k,p}=a_{p,k}^*$.
The brackets must vanish for any allowed choice of the state's coefficients:
\begin{equation}\label{Eq:Brackets}
a_{k,k}c_{k,L}c^*_{k,j'}+\abs{a_{k,p}}e^{i
\theta}c_{p,L}c^*_{k,j'}+\abs{a_{k,p}}e^{-i\theta}c_{k,L}c^*_{p,j'}+(1-a_{k,k})c_{p,L}c^*_{p,j'}=0.
\end{equation}
The above must also hold for $\theta\rightarrow \theta+\theta_1$ for any $\theta_1$:
\begin{equation}
a_{k,k}c_{k,L}c^*_{k,j'}+\abs{a_{k,p}}e^{i
\theta}e^{i\theta_1}c_{p,L}c^*_{k,j'}+\abs{a_{k,p}}e^{-i\theta}e^{-i\theta_1}c_{k,L}c^*_{p,j'}+(1-a_{k,k})c_{p,L}c^*_{p,j'}=0.
\end{equation}
Subtracting the second constraint equation from the first and dividing out $\abs{a_{k,p}}e^{i\theta}$,
\begin{equation}
\abs{a_{k,p}}e^{i\theta}(1-e^{i\theta_1})c_{p,L}c^*_{k,j'}+\abs{a_{k,p}}e^{-i\theta}(1-e^{-i\theta_1})c_{k,L}c^*_{p,j'}=0.
\end{equation}
Taking $\theta_1=\pi$ yields
\begin{equation}
c_{p,L}c^*_{k,j'}+e^{-2i\theta}c_{k,L}c^*_{p,j'}=0.
\end{equation}

Setting $\theta=0$ and $\theta=\pi/2$ yields
\begin{equation}
\begin{split}
c_{p,L}c^*_{k,j'}+c_{k,L}c^*_{p,j'}=0\\
c_{p,L}c^*_{k,j'}-c_{k,L}c^*_{p,j'}=0,
\end{split}
\end{equation}
requiring that $c_{p,L}c^*_{k,j'}=0$. Either $c_{p,L}=0$ for all $p\in \mathcal{A}$ or $c_{k,j'}=0$ for all $k,j'\in \mathcal{A}$. These conditions ensure that Eq.~\eqref{Eq:Brackets} is satisfied.

Now compute the third term:
\begin{equation}
\begin{split}
Q_2^\dagger \rho_{k,p} Q_1
	&=	\sum_{i'\in \mathcal{A}}c^*_{i',L}\ket{\psi_{L}}\bra{\psi_{i'}}\left[a_{k,k}\ket{\psi_k}\bra{\psi_k}+a_{k,p}\ket{\psi_k}\bra{\psi_p}+a_{p,k}\ket{\psi_p}\bra{\psi_k}+a_{p,p}\ket{\psi_p}\bra{\psi_p}\right]\sum_{i,j\in \mathcal{A}}c_{i,j}\ket{\psi_{i}}\bra{\psi_{j}}
\\
	&=	\sum_{i',i,j\in \mathcal{A}}c^*_{i',L}c_{i,j}\ket{\psi_{L}}\bra{\psi_{j}}
\left[a_{k,k}\delta_{i',k}\delta_{k,i}+a_{k,p}\delta_{i',k}\delta_{p,i}+a_{p,k}\delta_{i',p}\delta_{k,i}+a_{p,p}\delta_{i',p}\delta_{p,i}\right]\\
	&=	\sum_{j\in \mathcal{A}}\ket{\psi_{L}}\bra{\psi_{j}}
\left[a_{k,k}c^*_{k,L}c_{k,j}+a_{k,p}c^*_{k,L}c_{p,j}+a_{p,k}c^*_{p,L}c_{k,j}+a_{p,p}c^*_{p,L}c_{p,j}\right].
\end{split}
\end{equation}
The above conditions on the coefficients ensure that the brackets vanish.

Now compute, for $k,p\in \mathcal{A}$, 
\begin{equation}
\begin{split}
U \rho_{k,p} U^\dagger
	&=\left[\sum_{i',j'\in \mathcal{A}}c_{i',j'}\ket{\psi_{i'}}\bra{\psi_{j'}}+\sum_{i'\in \mathcal{B}}c_{i',l_{i'}}\ket{\psi_{i'}}\bra{\psi_{l_{i'}}}+\sum_{i'\in \mathcal{A}}c_{i',L}\ket{\psi_{i'}}\bra{\psi_{L}}+\sum_{i'\in \mathcal{A}}c_{m_{i'},i'}\ket{\psi_{m_{i'}}}\bra{\psi_{i'}}\right]\\
	&\hspace{20mm}\times\rho_{k,p}\left[\sum_{i,j\in \mathcal{A}}c^*_{i,j}\ket{\psi_j}\bra{\psi_i}+\sum_{i\in \mathcal{B}}c^*_{i,l_i}\ket{\psi_{l_i}}\bra{\psi_i}+\sum_{i\in \mathcal{A}}c^*_{i,L}\ket{\psi_{L}}\bra{\psi_i}+\sum_{i\in \mathcal{A}}c^*_{m_i,i}\ket{\psi_i}\bra{\psi_{m_i}}\right]\\
	&=\left[\sum_{i',j'\in \mathcal{A}}c_{i',j'}\ket{\psi_{i'}}\bra{\psi_{j'}}+\sum_{i'\in \mathcal{A}}c_{m_{i'},i'}\ket{\psi_{m_{i'}}}\bra{\psi_{i'}}\right]\rho_{k,p}\left[\sum_{i,j\in \mathcal{A}}c^*_{i,j}\ket{\psi_j}\bra{\psi_i}+\sum_{i\in \mathcal{A}}c^*_{m_i,i}\ket{\psi_i}\bra{\psi_{m_i}}\right]\\
		&=\left[Q_1+Q_3\right]\rho_{k,p}\left[Q_1^\dagger+Q_3^\dagger\right]\\
		&=Q_1\rho_{k,p}Q_1^\dagger+Q_1\rho_{k,p}Q_3^\dagger +Q_3\rho_{k,p}Q_1^\dagger +Q_3\rho_{k,p}Q_3^\dagger,
\end{split}
\end{equation}
where $Q_3=\sum_{i'\in \mathcal{A}}c_{m_{i'},i'}\ket{\psi_{m_{i'}}}\bra{\psi_{i'}}$.
We find the conditions on the coefficients which ensure that $U O_{k,p}U^\dagger \in \mathcal{O}_F$. We compute the fourth term in the above equation:

\begin{equation}
\begin{split}
Q_3 \rho_{k,p} Q_3^\dagger
	&=\sum_{i'\in \mathcal{A}}c_{m_{i'},i'}\ket{\psi_{m_{i'}}}\bra{\psi_{i'}}	\left[a_{k,k}\ket{\psi_k}\bra{\psi_k}+a_{k,p}\ket{\psi_k}\bra{\psi_p}+a_{p,k}\ket{\psi_p}\bra{\psi_k}+a_{p,p}\ket{\psi_p}\bra{\psi_p}\right]\sum_{i\in \mathcal{A}}c^*_{m_i,i}\ket{\psi_i}\bra{\psi_{m_i}}\\
	&=\sum_{i,i'\in \mathcal{A}}c_{m_{i'},i'}c^*_{m_i,i}\ket{\psi_{m_{i'}}}\bra{\psi_{m_i}}	\left[a_{k,k}\delta_{i',k}\delta_{k,i}+a_{k,p}\delta_{i',k}\delta_{p,i}+a_{p,k}\delta_{i',p}\delta_{k,i}+a_{p,p}\delta_{i',p}\delta_{p,i}\right]\\
	&=a_{k,k}c_{m_{k},k}c^*_{m_k,k}\ket{\psi_{m_{k}}}\bra{\psi_{m_k}}+a_{k,p}c_{m_{k},k}c^*_{m_p,p}\ket{\psi_{m_{k}}}\bra{\psi_{m_p}}\\
	&\hspace{10mm}+a_{p,k}c_{m_{p},p}c^*_{m_k,k}\ket{\psi_{m_{p}}}\bra{\psi_{m_k}}+a_{p,p}c_{m_{p},p}c^*_{m_p,p}\ket{\psi_{m_{p}}}\bra{\psi_{m_p}}.
\end{split}
\end{equation}

We require that $c_{m_k,k}c^*_{m_p,p}=0$ for $m_k\neq m_p$ and all $p,k\in \mathcal{A}$. This implies that there exists an $M\in \{m_k\}_{k\in\mathcal{A}}$ such that $c_{m_k,k}=0$ for $m_k\neq M$ and $k\in \mathcal{A}$. This yields $Q_3=\sum_{i'\in \mathcal{A}, m_{i'}\neq M}c_{M,i'}\ket{\psi_{M}}\bra{\psi_{i'}}$. By including some nonzero coefficients, we can write this as $Q_3=\sum_{i'\in \mathcal{A}}c_{M,i'}\ket{\psi_{M}}\bra{\psi_{i'}}$. Now we compute the second term:
\begin{equation}
\begin{split}
Q_1 \rho_{k,p} Q_3^\dagger
	&=\sum_{i',j'\in \mathcal{A}}c_{i',j'}\ket{\psi_{i'}}\bra{\psi_{j'}}\left[a_{k,k}\ket{\psi_k}\bra{\psi_k}+a_{k,p}\ket{\psi_k}\bra{\psi_p}+a_{p,k}\ket{\psi_p}\bra{\psi_k}+a_{p,p}\ket{\psi_p}\bra{\psi_p}\right]\sum_{i\in \mathcal{A}}c^*_{M,i}\ket{\psi_i}\bra{\psi_{M}}\\
	&=\sum_{i,i',j'\in \mathcal{A}}c_{i',j'}c^*_{M,i}\ket{\psi_{i'}}\bra{\psi_{M}}\left[a_{k,k}\delta_{j',k}\delta_{k,i}+a_{k,p}\delta_{j',k}\delta_{p,i}+a_{p,k}\delta_{j',p}\delta_{k,i}+a_{p,p}\delta_{j',p}\delta_{p,i}\right]\\
	&=\sum_{i'\in \mathcal{A}}\left[a_{k,k}c_{i',k}c^*_{M,k}+a_{k,p}c_{i',k}c^*_{M,p}+a_{p,k}c_{i',p}c^*_{M,k}+a_{p,p}c_{i',p}c^*_{M,p}\right]\ket{\psi_{i'}}\bra{\psi_{M}}.
\end{split}
\end{equation}
The brackets must vanish for any allowed choice of the state's coefficients. Similar to the proof above, one can show that $c_{i',k}c^*_{M,p}=0$ for all $i',p,k\in \mathcal{A}$. Either $c_{i',k}=0$ for all $i',k\in \mathcal{A}$ or $c_{M,p}=0$ for all $p\in \mathcal{A}$. We now compute the third term:

\begin{equation}
\begin{split}
Q_3 \rho_{k,p} Q_1^\dagger
	&=\sum_{i'\in \mathcal{A}}c_{M,i'}\ket{\psi_{M}}\bra{\psi_{i'}}	\left[a_{k,k}\ket{\psi_k}\bra{\psi_k}+a_{k,p}\ket{\psi_k}\bra{\psi_p}+a_{p,k}\ket{\psi_p}\bra{\psi_k}+a_{p,p}\ket{\psi_p}\bra{\psi_p}\right]\sum_{i,j\in \mathcal{A}}c^*_{i,j}\ket{\psi_{j}}\bra{\psi_{i}}\\
	&=\sum_{i,j,i'\in \mathcal{A}}c_{M,i'}c^*_{i,j}\ket{\psi_{M}}\bra{\psi_{i}}	\left[a_{k,k}\delta_{i',k}\delta_{k,j}+a_{k,p}\delta_{i',k}\delta_{p,j}+a_{p,k}\delta_{i',p}\delta_{k,j}+a_{p,p}\delta_{i',p}\delta_{p,j}\right]\\
	&=\sum_{i\in \mathcal{A}}\left[a_{k,k}c_{M,k}c^*_{i,k}+a_{k,p}c_{M,k}c^*_{i,p}+a_{p,k}c_{M,p}c^*_{i,k}+a_{p,p}c_{M,p}c^*_{i,p}\right]\ket{\psi_{M}}\bra{\psi_{i}}.
\end{split}
\end{equation}
The conditions found above ensure that the brackets vanish. 

Given all of the constraints on the coefficients, there are two possible forms of $U$:
\begin{equation}
	\begin{split}
	U_1&=\sum_{i,j\in \mathcal{A}}c_{i,j}\ket{\psi_i}\bra{\psi_j}+\sum_{i\in \mathcal{B}}c_{i,l_i}\ket{\psi_i}\bra{\psi_{l_i}},\\
	U_2&=\sum_{i\in \mathcal{B}}c_{i,l_i}\ket{\psi_i}\bra{\psi_{l_i}}+\sum_{i\in \mathcal{A}}c_{i,L}\ket{\psi_i}\bra{\psi_{L}}+\sum_{i\in \mathcal{A}}c_{M,i}\ket{\psi_{M}}\bra{\psi_i}.
	\end{split}
\end{equation} 
We show that $U_2$ is not a unitary. We compute 
\begin{equation}
\begin{split}
	U_2^\dagger U_2
	&=\left[\sum_{i'\in \mathcal{B}}c^*_{i',l_{i'}}\ket{\psi_{l_{i'}}}\bra{\psi_{i'}}+\sum_{i'\in \mathcal{A}}c^*_{i',L}\ket{\psi_{L}}\bra{\psi_{i'}}+\sum_{i'\in \mathcal{A}}c^*_{M,i'}\ket{\psi_{i'}}\bra{\psi_{M}}\right]\\
	&\hspace{5mm}\times\left[\sum_{i\in \mathcal{B}}c_{i,l_i}\ket{\psi_i}\bra{\psi_{l_i}}+\sum_{i\in \mathcal{A}}c_{i,L}\ket{\psi_i}\bra{\psi_{L}}+\sum_{i\in \mathcal{A}}c_{M,i}\ket{\psi_{M}}\bra{\psi_i}\right]\\
	&=\sum_{i'\in \mathcal{B}}c^*_{i',l_{i'}}\ket{\psi_{l_{i'}}}\bra{\psi_{i'}}\left[\sum_{i\in \mathcal{B}}c_{i,l_i}\ket{\psi_i}\bra{\psi_{l_i}}+\sum_{i\in \mathcal{A}}c_{M,i}\ket{\psi_{M}}\bra{\psi_i}\right]+\sum_{i,i'\in \mathcal{A}}c^*_{i',L}c_{i,L}\ket{\psi_{L}}\bra{\psi_{i'}}\psi_i\rangle\bra{\psi_{L}}\\
	&\hspace{5mm}+\sum_{i'\in \mathcal{A}}c^*_{M,i'}\ket{\psi_{i'}}\bra{\psi_{M}}\left[\sum_{i\in \mathcal{B}}c_{i,l_i}\ket{\psi_i}\bra{\psi_{l_i}}+\sum_{i\in \mathcal{A}}c_{M,i}\ket{\psi_{M}}\bra{\psi_i}\right]\\
	&=\sum_{i,i'\in \mathcal{B}}c^*_{i',l_{i'}}c_{i,l_i}\ket{\psi_{l_{i'}}}\bra{\psi_{l_i}}\delta_{i,i'}+\sum_{\substack{i\in \mathcal{A},\\i'\in \mathcal{B}}}c^*_{i',l_{i'}}c_{M,i}\ket{\psi_{l_{i'}}}\bra{\psi_i}\delta_{i',M}+\sum_{i,i'\in \mathcal{A}}c^*_{i',L}c_{i,L}\ket{\psi_{L}}\bra{\psi_{L}}\delta_{i,i'}\\
	&\hspace{5mm}+\sum_{\substack{i'\in \mathcal{A},\\i\in \mathcal{B}}}c^*_{M,i'}c_{i,l_i}\ket{\psi_{i'}}\bra{\psi_{l_i}}\delta_{i,M}+\sum_{i,i'\in \mathcal{A}}c^*_{M,i'}c_{M,i}\ket{\psi_{i'}}\bra{\psi_i}\\
	&=\sum_{i\in \mathcal{B}}c^*_{i,l_{i}}c_{i,l_i}\ket{\psi_{l_{i}}}\bra{\psi_{l_i}}+\sum_{\substack{i\in A}}c^*_{M,l_{M}}c_{M,i}\ket{\psi_{l_{M}}}\bra{\psi_i}+\sum_{i\in A}\abs{c_{i,J}}^2\ket{\psi_{L}}\bra{\psi_{L}}\\
	&\hspace{5mm}+\sum_{\substack{i'\in \mathcal{A}}}c^*_{M,i'}c_{M,l_M}\ket{\psi_{i'}}\bra{\psi_{l_M}}+\sum_{i,i'\in \mathcal{A}}c^*_{M,i'}c_{M,i}\ket{\psi_{i'}}\bra{\psi_i}.
\end{split}
\end{equation}
Unitarity requires that $U_2^\dagger U_2=I=\sum_{i\in \mathcal{A}\cup \mathcal{B}}\ket{\psi_i}\bra{\psi_i}$. Therefore, we require that
\begin{equation}
	c^*_{M,i'}c_{M,i}=\delta_{i,i'}
\end{equation}
for all $i,i'\in \mathcal{A}$. However, no such $c_{M,i}$ exists, so $U_2$ is not unitary and $U_2\notin \mathcal{U}_{\mathcal{M}_F}$.

We now find the conditions that make $U_1$ unitary:
\begin{equation}
\begin{split}
U_1^\dagger U_1
	&=\left[\sum_{i',j'\in \mathcal{A}}c^*_{i',j'}\ket{\psi_{j'}}\bra{\psi_{i'}}+\sum_{i'\in \mathcal{B}}c^*_{i',l_{i'}}\ket{\psi_{l_{i'}}}\bra{\psi_{i'}}\right]\left[\sum_{i,j\in \mathcal{A}}c_{i,j}\ket{\psi_i}\bra{\psi_j}+\sum_{i\in \mathcal{B}}c_{i,l_i}\ket{\psi_i}\bra{\psi_{l_i}}\right]\\
	&=\sum_{i,j,i',j'\in \mathcal{A}}c^*_{i',j'}c_{i,j}\ket{\psi_{j'}}\bra{\psi_{i'}} \psi_i\rangle\bra{\psi_j}+\sum_{i,i'\in \mathcal{B}}c^*_{i',l_{i'}}c_{i,l_i}\ket{\psi_{l_{i'}}}\bra{\psi_{i'}}\psi_i\rangle\bra{\psi_{l_i}}\\
	&=\sum_{i,j,i',j'\in \mathcal{A}}c^*_{i',j'}c_{i,j}\ket{\psi_{j'}}\bra{\psi_j}\delta_{i,i'}+\sum_{i,i'\in \mathcal{B}}c^*_{i',l_{i'}}c_{i,l_i}\ket{\psi_{l_{i'}}}\bra{\psi_{l_i}}\delta_{i,i'}\\
	&=\sum_{i,j,j'\in \mathcal{A}}c^*_{i,j'}c_{i,j}\ket{\psi_{j'}}\bra{\psi_j}+\sum_{i\in \mathcal{B}}\abs{c_{i,l_i}}^2\ket{\psi_{l_{i}}}\bra{\psi_{l_i}}.
\end{split}
\end{equation}
Unitarity requires that 
 \begin{equation}\label{Eq:li}
 	\{l_i\}_{i\in\mathcal{B}}=\mathcal{B},
 \end{equation}
 \begin{equation}\label{Eq:cli}
 	\abs{c_{i,l_i}}^2=1, \ \forall i\in \mathcal{B},
 \end{equation}
 \begin{equation}
 	\sum_{i\in \mathcal{A}}c^*_{i,j'}c_{i,j}=\delta_{j',j}, \quad j,j'\in \mathcal{A}.
 \end{equation}
The first equation implies that $l_i$ can be written as $l_i=\sigma(i)$ where $\sigma$ is a permutation of $\mathcal{B}$.

We ensure that the conditions previously found for $c_{i,l_i}$ are consistent with Eqs.~\eqref{Eq:li} and \eqref{Eq:cli}. For $k\in \mathcal{B}$, recall that there exists an index $m_k\in \mathcal{B}$ such that $c_{i,k}=0$ for $i\in \mathcal{B}$ satisfying $i\neq m_k$. Also, there exists an index $l_k\in \mathcal{B}$ such that $c_{k,j}=0$ for $j\in \mathcal{B}$ and $j\neq l_k$. Assume that, for $k_1\in \mathcal{B}$, $l_{k_1}=l'$. There exists an index $m_{l'}\in \mathcal{B}$ such that $c_{i,l'}=0$ for $i\in \mathcal{B}$ and $i\neq m_{l'}$. If $k_1\neq m_{l'}$, then $c_{k_1,l_{k_1}}=0$, contradicting Eq.~\eqref{Eq:cli}. Therefore, $m_{l'}= k_1$. In other words, if $k_1,k_2\in\mathcal{B}$ and $l_{k_1}=l_{k_2}$, then $k_1=k_2$. This ensures that $\{l_k\}_{k\in \mathcal{B}}=\mathcal{B}$, consistent with Eq.~\eqref{Eq:li}.

Lastly, for an arbitrary free pure state $\rho=\sum_{k,p\in \mathcal{A}}a_{k,p}\ket{\psi_k}\bra{\psi_p}+\sum_{k\in \mathcal{B}}a_{k,k}\ket{\psi_k}\bra{\psi_k}$, we verify that $U_1^\dagger \rho  U_1\in \mathcal{M}_F$:
\begin{equation}
\begin{split}
	U_1^\dagger \rho U_1
	&=\left[\sum_{i',j'\in \mathcal{A}}c^*_{i',j'}\ket{\psi_{j'}}\bra{\psi_{i'}}+\sum_{i'\in \mathcal{B}}c^*_{i',l_{i'}}\ket{\psi_{l_{i'}}}\bra{\psi_{i'}}\right]\left[\sum_{k,p\in \mathcal{A}}a_{k,p}\ket{\psi_k}\bra{\psi_p}+\sum_{k\in \mathcal{B}}a_{k,k}\ket{\psi_k}\bra{\psi_k}\right]\\
	&\hspace{5mm}\times\left[\sum_{i,j\in \mathcal{A}}c_{i,j}\ket{\psi_i}\bra{\psi_j}+\sum_{i\in \mathcal{B}}c_{i,l_i}\ket{\psi_i}\bra{\psi_{l_i}}\right]\\
	&=\sum_{i,j,i',j',k,p\in \mathcal{A}}c^*_{i',j'}c_{i,j}a_{k,p}\ket{\psi_{j'}}\bra{\psi_{j}}\delta_{i',k}\delta_{p,i}+\sum_{i,i',k\in \mathcal{B}}c^*_{i',l_{i'}}c_{i,l_i}a_{k,k}\ket{\psi_{l_{i'}}}\bra{\psi_{l_i}}\delta_{i',k}\delta_{k,i}\\
	&=\sum_{j,j',k,p\in \mathcal{A}}c^*_{k,j'}c_{p,j}a_{k,p}\ket{\psi_{j'}}\bra{\psi_{j}}+\sum_{k\in \mathcal{B}}c^*_{k,l_{k}}c_{k,l_k}a_{k,k}\ket{\psi_{l_{k}}}\bra{\psi_{l_k}}\\
	&=\sum_{j,j'\in \mathcal{A}}\left[\sum_{k,p\in \mathcal{A}}c^*_{k,j'}c_{p,j}a_{k,p}\right]\ket{\psi_{j'}}\bra{\psi_{j}}+\sum_{k\in \mathcal{B}}a_{k,k}\ket{\psi_{l_{k}}}\bra{\psi_{l_k}}\\
	&\in \mathcal{M}_F.
\end{split}
\end{equation}
 
Hence, all unitaries in $\mathcal{U}_{\mathcal{M}_F}$ have the form of $U_1$, implying that $\mathcal{U}_{\mathcal{M}_F}=\mathcal{U}_F$.

\section{Proof of Proposition~\ref{Prop:FreeObservable}}
Let $O$ be a normalized observable satisfying 
\begin{equation}
	O=\sum_{\substack{i,j \in  \mathcal{A}\\}}a_{i,j}\ket{\psi_i}\bra{\psi_j}+\sum_{i\in\mathcal{ B}}a_{i,i}\ket{\psi_i}\bra{\psi_i}.
\end{equation}
Then 
\begin{equation}
\begin{split}
	O(2\pi T)&=\sum_{\substack{i,j \in  \mathcal{A}\\}}a_{i,j}e^{-i(E_i-E_j)2\pi T}\ket{\psi_i}\bra{\psi_j}+\sum_{i\in \mathcal{B}}a_{i,i}e^{i(E_i-E_i)2\pi T}\ket{\psi_i}\bra{\psi_i}\\
	&=\sum_{\substack{i,j \in  \mathcal{A}\\}}a_{i,j}\ket{\psi_i}\bra{\psi_j}+\sum_{i\in \mathcal{B}}a_{i,i}\ket{\psi_i}\bra{\psi_i}\\
	&=O
\end{split}
\end{equation}
and $G(O)=1$.

Let $O$ be a normalized observable with the form
\begin{equation}
	O=\sum_{\substack{i,j\in \mathcal{B}_2\subseteq \mathcal{B},\\ i \neq j, \abs{\mathcal{B}_2}\geq 2}}a_{i,j}\ket{\psi_i}\bra{\psi_j},
\end{equation}
where all coefficients are nonzero. By the Hermiticity condition, this can be written as
\begin{equation}
	O=\sum_{\substack{i,j\in \mathcal{B}_2\subseteq \mathcal{B},\\ i \neq j,\abs{\mathcal{B}_2}\geq 2,\\ i<j}}a_{i,j}\ket{\psi_i}\bra{\psi_j}+a_{j,i}\ket{\psi_j}\bra{\psi_i}.
\end{equation}
Then
\begin{equation}
\begin{split}
	O(2\pi T)
	&=\sum_{\substack{i,j\in \mathcal{B}_2\subseteq \mathcal{B},\\ i \neq j,\abs{\mathcal{B}_2}\geq 2,\\ i<j}}a_{i,j}e^{-i(E_i-E_j)2\pi T}\ket{\psi_i}\bra{\psi_j}+a_{j,i}e^{i(E_i-E_j)2\pi T}\ket{\psi_j}\bra{\psi_i}\\
	&=\sum_{\substack{i,j\in \mathcal{B}_2\subseteq \mathcal{B},\\ i \neq j,\abs{\mathcal{B}_2}\geq 2,\\ i<j}}e^{-i(E_i-E_j)2\pi T}\left[a_{i,j}\ket{\psi_i}\bra{\psi_j}+a_{j,i}e^{i(E_i-E_j)4\pi T}\ket{\psi_j}\bra{\psi_i}\right].
\end{split}
\end{equation}
Since, by assumption, $E_i-E_j$ is not a rational number, then $e^{i(E_i-E_j)4\pi T}\neq 1$ and $O\neq e^{i\phi} O(2\pi T)$ for any phase $\phi$. Therefore $G(O)<1$. 

Let $O$ be a normalized observable with the form
\begin{equation}
	O=\sum_{\substack{i\in \mathcal{A}'\subseteq \mathcal{A},j\in \mathcal{B}_1\subseteq \mathcal{B},\\  \abs{\mathcal{A}'}\geq 1, \abs{\mathcal{B}_1}\geq 1}}a_{i,j}\ket{\psi_i}\bra{\psi_j}+a_{j,i}\ket{\psi_j}\bra{\psi_i}
\end{equation}
where all coefficients are nonzero. Then
\begin{equation}
\begin{split}
	O(2\pi T)
	&=\sum_{\substack{i\in \mathcal{A}'\subseteq \mathcal{A},j\in \mathcal{B}_1\subseteq \mathcal{B},\\  \abs{\mathcal{A}'}\geq 1, \abs{\mathcal{B}_1}\geq 1}}a_{i,j}e^{-i(E_i-E_j)2\pi T}\ket{\psi_i}\bra{\psi_j}+a_{j,i}e^{i(E_i-E_j)2\pi T}\ket{\psi_j}\bra{\psi_i}\\
	&=\sum_{\substack{i\in \mathcal{A}'\subseteq \mathcal{A},j\in \mathcal{B}_1\subseteq \mathcal{B},\\  \abs{\mathcal{A}'}\geq 1, \abs{\mathcal{B}_1}\geq 1}}a_{i,j}e^{iE_j2\pi T}\ket{\psi_i}\bra{\psi_j}+a_{j,i}e^{-iE_j2\pi T}\ket{\psi_j}\bra{\psi_i}.
\end{split}
\end{equation}
Since $E_j$ is an irrational number, then $e^{i\phi}O(2\pi T)\neq O$ for any phase $\phi$. Hence, $G(O)<1$. In general, a normalized observable $O$ can be written as
\begin{equation}
	O=\sum_{\substack{i,j \in  \mathcal{A}\\}}a_{i,j}\ket{\psi_i}\bra{\psi_j}+\sum_{i\in \mathcal{B}}a_{i,i}\ket{\psi_i}\bra{\psi_i}+\sum_{i,j\in \mathcal{B}, i\neq j}a_{i,j}\ket{\psi_i}\bra{\psi_j}+\sum_{i\in \mathcal{A},j\in \mathcal{B}}a_{i,j}\ket{\psi_i}\bra{\psi_j}+\sum_{j\in \mathcal{A},i\in \mathcal{B}}a_{i,j}\ket{\psi_i}\bra{\psi_j}.
\end{equation}
If there exist any nonzero coefficients in the last three summations, then $G(O)<1$. Therefore, $G(O)=1$ if and only if
\begin{equation}
	O=\sum_{\substack{i,j \in  \mathcal{A}\\}}a_{i,j}\ket{\psi_i}\bra{\psi_j}+\sum_{i\in \mathcal{B}}a_{i,i}\ket{\psi_i}\bra{\psi_i}.
\end{equation}

\section{Free unitaries for observables}
We define $\mathcal{U}_{\mathcal{O}_F}$ as the set of unitaries which map $\mathcal{O}_F$ to itself. We will show that $\mathcal{U}_{\mathcal{O}_F}=\mathcal{U}_F$. We first take a unitary $U\in \mathcal{U}_{\mathcal{O}_F}$ to have the general form
\begin{equation}
U=\sum_{i,j\in \mathcal{A}\cup \mathcal{B}}c_{i,j}\ket{\psi_i}\bra{\psi_j}.
\end{equation}
We define $O_k=a_{k,k}\ket{\psi_k}\bra{\psi_k}\in \mathcal{O}_F$, where $k\in \mathcal{A} \cup \mathcal{B}$ and $a_k$ is chosen to satisfy $\norm{O_k}_2=1$. We compute the following
\begin{equation}
\begin{split}
U^\dagger O_k U
	&=\left[	\sum_{i',j'\in \mathcal{A}\cup \mathcal{B}}c^*_{i',j'}\ket{\psi_{j'}}\bra{\psi_{i'}}\right]\left[a_{k,k}\ket{\psi_k}\bra{\psi_k}\right] \left[\sum_{i,j\in \mathcal{A}\cup\mathcal{B}}c_{i,j}\ket{\psi_i}\bra{\psi_j}\right]\\
	&=a_{k,k}\sum_{i',j',i,j\in \mathcal{A}\cup \mathcal{B}}c^*_{i',j'}c_{i,j}\ket{\psi_{j'}}\bra{\psi_j}\delta_{i',k}\delta_{k,i}\\
	&=a_{k,k}\sum_{j',j\in \mathcal{A}\cup \mathcal{B}}c^*_{k,j'}c_{k,j}\ket{\psi_{j'}}\bra{\psi_j}.
\end{split}
\end{equation}
Since $U^\dagger O_k U\in \mathcal{O}_F$, it follows that $c^*_{k,j'}c_{k,j}=0$ for $k\in \mathcal{A}\cup \mathcal{B}$ and $j,j'\in \mathcal{B}$ where $j\neq j'$. 
This implies that there exists an index $l_k\in \mathcal{B}$ such that $c_{k,j}=0$ for $j\in \mathcal{B}$ and $j\neq l_k$. The subscript in $l_k$ denotes that this index may depend on the value of $k$.

Before proceeding with the proof, we introduce the following lemma.
\begin{lemma}\label{Lemma:FreeUnitaryObs}
If $U\in \mathcal{U}_{\mathcal{O}_F}$, then $U^\dagger \in \mathcal{U}_{\mathcal{O}_F}$.
\end{lemma}
\begin{proof}
By definition, for a unitary $U\in\mathcal{U}_{\mathcal{O}_F}$,  $\{U^\dagger O U : O\in \mathcal{O}_F\}\subseteq \mathcal{O}_F$. Unitaries are one-to-one mappings, implying that for any two free observables $O_1$ and $O_2$, if $O_1\neq O_2$, then $U^\dagger O_1 U \neq U^\dagger O_2 U$. Therefore, $\{U^\dagger O_1 U : O_1\in \mathcal{O}_F\}=\mathcal{O}_F$. For any fixed free observable $O_2$, there exists a free observable $O_1$ such that $O_2=U^\dagger O_1 U$. Hence $UO_2U^\dagger=O_1\in \mathcal{O}_F$ for all $O_2\in\mathcal{O}_F$ and $U^\dagger \in \mathcal{U}_{\mathcal{O}_F}$.
\end{proof}
Since $U^\dagger\in \mathcal{U}_{\mathcal{O}_F}$ by Lemma~\ref{Lemma:FreeUnitaryObs}, then $U O_k U^\dagger \in\mathcal{O}_F$. We compute 
\begin{equation}
\begin{split}
U O_k U^\dagger
	&=\left[	\sum_{i',j'\in \mathcal{A}\cup \mathcal{B}}c_{i',j'}\ket{\psi_{i'}}\bra{\psi_{j'}}\right]\left[a_{k,k}\ket{\psi_k}\bra{\psi_k}\right] \left[\sum_{i,j\in \mathcal{A}\cup \mathcal{B}}c^*_{i,j}\ket{\psi_j}\bra{\psi_i}\right]\\
	&=a_{k,k}\sum_{i',j',i,j\in \mathcal{A}\cup \mathcal{B}}c_{i',j'}c^*_{i,j}\ket{\psi_{i'}}\bra{\psi_i}\delta_{j',k}\delta_{k,j}\\
	&=a_{k,k}\sum_{i',i\in \mathcal{A}\cup \mathcal{B}}c_{i',k}c^*_{i,k}\ket{\psi_{i'}}\bra{\psi_i}.
\end{split}
\end{equation}
It follows that $c_{i',k}c^*_{i,k}=0$ for $k\in \mathcal{A}\cup \mathcal{B}$ and $i,i'\in \mathcal{B}$ where $i\neq i'$. This implies that there exists an index $m_k\in \mathcal{B}$ such that $c_{i,k}=0$ for $i\in \mathcal{B}$ and $i\neq m_k$. The unitary becomes
\begin{equation}
	\begin{split}
	U
	&=\sum_{i,j\in \mathcal{A}\cup \mathcal{B}}c_{i,j}\ket{\psi_i}\bra{\psi_j}\\
	&=\sum_{i,j\in \mathcal{A}}c_{i,j}\ket{\psi_i}\bra{\psi_j}+\sum_{i,j\in \mathcal{B}}c_{i,j}\ket{\psi_i}\bra{\psi_j}+\sum_{i\in \mathcal{A},j\in \mathcal{B}}c_{i,j}\ket{\psi_i}\bra{\psi_j}+\sum_{i\in \mathcal{B},j\in \mathcal{A}}c_{i,j}\ket{\psi_i}\bra{\psi_j}\\
	&=\sum_{i,j\in \mathcal{A}}c_{i,j}\ket{\psi_i}\bra{\psi_j}+\sum_{i\in \mathcal{B}}c_{i,l_i}\ket{\psi_i}\bra{\psi_{l_i}}+\sum_{i\in \mathcal{A}}c_{i,l_i}\ket{\psi_i}\bra{\psi_{l_i}}+\sum_{i\in \mathcal{A}}c_{m_i,i}\ket{\psi_{m_i}}\bra{\psi_i},
	\end{split}
\end{equation} 
where $l_i,m_i\in \mathcal{B}$.

Define the normalized observable
\begin{equation}
	O_{k,p}=a_{k,p}\ket{\psi_k}\bra{\psi_p}+a_{p,k}\ket{\psi_p}\bra{\psi_k},
\end{equation}
where $k,p\in \mathcal{A}$.
Now compute
\begin{equation}
\begin{split}
U^\dagger O_{k,p} U
	&=\left[\sum_{i',j'\in \mathcal{A}}c^*_{i',j'}\ket{\psi_{j'}}\bra{\psi_{i'}}+\sum_{i'\in \mathcal{B}}c^*_{i',l_{i'}}\ket{\psi_{l_{i'}}}\bra{\psi_{i'}}+\sum_{i'\in \mathcal{A}}c^*_{i',l_{i'}}\ket{\psi_{l_{i'}}}\bra{\psi_{i'}}+\sum_{i'\in \mathcal{A}}c^*_{m_{i'},i'}\ket{\psi_{i'}}\bra{\psi_{m_{i'}}}\right]\\
	&\hspace{20mm}\times\left[a_{k,p}\ket{\psi_k}\bra{\psi_p}+a_{p,k}\ket{\psi_p}\bra{\psi_k}\right] \\
	&\hspace{20mm}\times\left[\sum_{i,j\in \mathcal{A}}c_{i,j}\ket{\psi_i}\bra{\psi_j}+\sum_{i\in \mathcal{B}}c_{i,l_i}\ket{\psi_i}\bra{\psi_{l_i}}+\sum_{i\in \mathcal{A}}c_{i,l_i}\ket{\psi_i}\bra{\psi_{l_i}}+\sum_{i\in \mathcal{A}}c_{m_i,i}\ket{\psi_{m_i}}\bra{\psi_i}\right]\\
	&=\left[\sum_{i',j'\in \mathcal{A}}c^*_{i',j'}\ket{\psi_{j'}}\bra{\psi_{i'}}+\sum_{i'\in \mathcal{A}}c^*_{i',l_{i'}}\ket{\psi_{l_{i'}}}\bra{\psi_{i'}}\right]
	\left[a_{k,p}\ket{\psi_k}\bra{\psi_p}+a_{p,k}\ket{\psi_p}\bra{\psi_k}\right] \\
	&\hspace{20mm}\times\left[\sum_{i,j\in \mathcal{A}}c_{i,j}\ket{\psi_i}\bra{\psi_j}+\sum_{i\in \mathcal{A}}c_{i,l_i}\ket{\psi_i}\bra{\psi_{l_i}}\right]\\
	&=\left[Q_1^\dagger +Q_2^\dagger \right]O_{k,p}\left[Q_1 +Q_2 \right]\\
	&=Q_1^\dagger O_{k,p}Q_1+Q_1^\dagger O_{k,p}Q_2+Q_2^\dagger O_{k,p}Q_1+Q_2^\dagger O_{k,p}Q_2 ,
\end{split}
\end{equation}
where we define $Q_1=\sum_{i,j\in \mathcal{A}}c_{i,j}\ket{\psi_i}\bra{\psi_j}$ and $Q_2=\sum_{i\in \mathcal{A}}c_{i,l_i}\ket{\psi_i}\bra{\psi_{l_i}}$.
We now prove conditions on the coefficients to ensure that $U^\dagger O_{k,p}U\in \mathcal{O}_F$. We compute the fourth term in the above equation:

\begin{equation}
\begin{split}
Q_2^\dagger O_{k,p} Q_2
	&=\sum_{i'\in \mathcal{A}}c^*_{i',l_{i'}}\ket{\psi_{l_{i'}}}\bra{\psi_{i'}}
	\left[a_{k,p}\ket{\psi_k}\bra{\psi_p}+a_{p,k}\ket{\psi_p}\bra{\psi_k}\right]\sum_{i\in \mathcal{A}}c_{i,l_i}\ket{\psi_i}\bra{\psi_{l_i}}\\
	&=\sum_{i,i'\in \mathcal{A}}c_{i,l_i}c^*_{i',l_{i'}}\ket{\psi_{l_{i'}}}\bra{\psi_{l_i}}
	\left[a_{k,p}\delta_{i',k}\delta_{p,i}+a_{p,k}\delta_{i',p}\delta_{k,i}\right]\\
	&=a_{k,p}c_{p,l_p}c^*_{k,l_{k}}\ket{\psi_{l_{k}}}\bra{\psi_{l_p}}
	+a_{p,k}c_{k,l_k}c^*_{p,l_{p}}\ket{\psi_{l_{p}}}\bra{\psi_{l_k}}.
\end{split}
\end{equation}
As $l_p,l_k\in\mathcal{B}$, we require that $c_{p,l_p}c^*_{k,l_{k}}=0$ for $l_p\neq l_k$ and for all $p,k\in \mathcal{A}$. This implies that there exists an index $L\in\{l_p\}_{p\in A}$ such that $c_{p,l_p}=0$ for $l_p\neq L$ and $p\in A$. This yields $Q_2=\sum_{i\in \mathcal{A},l_i=L}c_{i,L}\ket{\psi_i}\bra{\psi_{L}}$. We are free to include the terms in the sum where $l_i\neq L$, since $c_{i,l_i}=0$ in this case. We do this for convenience and write $Q_2=\sum_{i\in \mathcal{A}}c_{i,L}\ket{\psi_i}\bra{\psi_{L}}$.

Now we compute the second term:
\begin{equation}
\begin{split}
Q_1^\dagger O_{k,p} Q_2
	&=\sum_{i',j'\in \mathcal{A}}c^*_{i',j'}\ket{\psi_{j'}}\bra{\psi_{i'}}
	\left[a_{k,p}\ket{\psi_k}\bra{\psi_p}+a_{p,k}\ket{\psi_p}\bra{\psi_k}\right]\sum_{i\in \mathcal{A}}c_{i,L}\ket{\psi_i}\bra{\psi_{L}}\\
	&=\sum_{i,i',j'\in \mathcal{A}}c_{i,L}c^*_{i',j'}\ket{\psi_{j'}}\bra{\psi_{L}}
	\left[a_{k,p}\delta_{i',k}\delta_{p,i}+a_{p,k}\delta_{i',p}\delta_{k,i}\right]\\
	&=\sum_{j'\in \mathcal{A}}\ket{\psi_{j'}}\bra{\psi_{L}}
	\left[a_{k,p}c_{p,L}c^*_{k,j'}+a_{p,k}c_{k,L}c^*_{p,j'}\right].
\end{split}
\end{equation}
The brackets must vanish for any allowed choice of the coefficients $a_{k,p}$ and $a_{p,k}$. This is true only if $c_{p,L}c^*_{k,j'}=0$ and $c_{k,L}c^*_{p,j'}=0$ for all $p,k,j'\in \mathcal{A}$. Either $c_{p,L}=0$ for all $p\in \mathcal{A}$ or $c_{k,j'}=0$ for all $k,j'\in \mathcal{A}$.

Now compute the third term:
\begin{equation}
\begin{split}
Q_2^\dagger O_{k,p} Q_1
	&=	\sum_{i'\in \mathcal{A}}c^*_{i',L}\ket{\psi_{L}}\bra{\psi_{i'}}\left[a_{k,p}\ket{\psi_k}\bra{\psi_p}+a_{p,k}\ket{\psi_p}\bra{\psi_k}\right]\sum_{i,j\in \mathcal{A}}c_{i,j}\ket{\psi_{i}}\bra{\psi_{j}}
\\
	&=	\sum_{i',i,j\in \mathcal{A}}c^*_{i',L}c_{i,j}\ket{\psi_{L}}\bra{\psi_{j}}
\left[a_{k,p}\delta_{i',k}\delta_{p,i}+a_{p,k}\delta_{i',p}\delta_{k,i}\right]\\
	&=	\sum_{j\in \mathcal{A}}\ket{\psi_{L}}\bra{\psi_{j}}
\left[a_{k,p}c^*_{k,L}c_{p,j}+a_{p,k}c^*_{p,L}c_{k,j}\right].
\end{split}
\end{equation}
The above conditions on the coefficients ensure that the brackets vanish.

Now compute, for $k,p\in \mathcal{A}$, 
\begin{equation}
\begin{split}
U O_{k,p} U^\dagger
	&=\left[\sum_{i',j'\in \mathcal{A}}c_{i',j'}\ket{\psi_{i'}}\bra{\psi_{j'}}+\sum_{i'\in \mathcal{B}}c_{i',l_{i'}}\ket{\psi_{i'}}\bra{\psi_{l_{i'}}}+\sum_{i'\in \mathcal{A}}c_{i',L}\ket{\psi_{i'}}\bra{\psi_{L}}+\sum_{i'\in \mathcal{A}}c_{m_{i'},i'}\ket{\psi_{m_{i'}}}\bra{\psi_{i'}}\right]\\
	&\hspace{20mm}\times\left[a_{k,p}\ket{\psi_k}\bra{\psi_p}+a_{p,k}\ket{\psi_p}\bra{\psi_k}\right] \\
	&\hspace{20mm}\times\left[\sum_{i,j\in \mathcal{A}}c^*_{i,j}\ket{\psi_j}\bra{\psi_i}+\sum_{i\in \mathcal{B}}c^*_{i,l_i}\ket{\psi_{l_i}}\bra{\psi_i}+\sum_{i\in \mathcal{A}}c^*_{i,L}\ket{\psi_{L}}\bra{\psi_i}+\sum_{i\in \mathcal{A}}c^*_{m_i,i}\ket{\psi_i}\bra{\psi_{m_i}}\right]\\
	&=\left[\sum_{i',j'\in \mathcal{A}}c_{i',j'}\ket{\psi_{i'}}\bra{\psi_{j'}}+\sum_{i'\in \mathcal{A}}c_{m_{i'},i'}\ket{\psi_{m_{i'}}}\bra{\psi_{i'}}\right]
	\left[a_{k,p}\ket{\psi_k}\bra{\psi_p}+a_{p,k}\ket{\psi_p}\bra{\psi_k}\right] \\
	&\hspace{20mm}\times\left[\sum_{i,j\in \mathcal{A}}c^*_{i,j}\ket{\psi_j}\bra{\psi_i}+\sum_{i\in \mathcal{A}}c^*_{m_i,i}\ket{\psi_i}\bra{\psi_{m_i}}\right]\\
		&=\left[Q_1+Q_3\right]O_{k,p}\left[Q_1^\dagger+Q_3^\dagger\right]\\
		&=Q_1O_{k,p}Q_1^\dagger+Q_1O_{k,p}Q_3^\dagger +Q_3O_{k,p}Q_1^\dagger +Q_3O_{k,p}Q_3^\dagger,
\end{split}
\end{equation}
where $Q_3=\sum_{i'\in \mathcal{A}}c_{m_{i'},i'}\ket{\psi_{m_{i'}}}\bra{\psi_{i'}}$.
We find the conditions on the coefficients which ensure that $U O_{k,p}U^\dagger \in \mathcal{O}_F$. We compute the fourth term in the above equation:
\begin{equation}
\begin{split}
Q_3 O_{k,p} Q_3^\dagger
	&=\sum_{i'\in \mathcal{A}}c_{m_{i'},i'}\ket{\psi_{m_{i'}}}\bra{\psi_{i'}}	\left[a_{k,p}\ket{\psi_k}\bra{\psi_p}+a_{p,k}\ket{\psi_p}\bra{\psi_k}\right]\sum_{i\in \mathcal{A}}c^*_{m_i,i}\ket{\psi_i}\bra{\psi_{m_i}}\\
	&=\sum_{i,i'\in \mathcal{A}}c_{m_{i'},i'}c^*_{m_i,i}\ket{\psi_{m_{i'}}}\bra{\psi_{m_i}}	\left[a_{k,p}\delta_{i',k}\delta_{p,i}+a_{p,k}\delta_{i',p}\delta_{k,i}\right]\\
	&=a_{k,p}c_{m_{k},k}c^*_{m_p,p}\ket{\psi_{m_{k}}}\bra{\psi_{m_p}}+a_{p,k}c_{m_{p},p}c^*_{m_k,k}\ket{\psi_{m_{p}}}\bra{\psi_{m_k}}.
\end{split}
\end{equation}
We require that $c_{m_k,k}c^*_{m_p,p}=0$ for $m_k\neq m_p$ and all $p,k\in \mathcal{A}$. This implies that there exists an $M\in \{m_k\}_{k\in\mathcal{A}}$ such that $c_{m_k,k}=0$ for $m_k\neq M$ and $k\in \mathcal{A}$. This yields $Q_3=\sum_{i'\in \mathcal{A}, m_{i'}\neq M}c_{M,i'}\ket{\psi_{M}}\bra{\psi_{i'}}$. By including some nonzero coefficients, we can write this as $Q_3=\sum_{i'\in \mathcal{A}}c_{M,i'}\ket{\psi_{M}}\bra{\psi_{i'}}$. Now we compute the second term:
\begin{equation}
\begin{split}
Q_1 O_{k,p} Q_3^\dagger
	&=\sum_{i',j'\in \mathcal{A}}c_{i',j'}\ket{\psi_{i'}}\bra{\psi_{j'}}\left[a_{k,p}\ket{\psi_k}\bra{\psi_p}+a_{p,k}\ket{\psi_p}\bra{\psi_k}\right]\sum_{i\in \mathcal{A}}c^*_{M,i}\ket{\psi_i}\bra{\psi_{M}}\\
	&=\sum_{i,i',j'\in \mathcal{A}}c_{i',j'}c^*_{M,i}\ket{\psi_{i'}}\bra{\psi_{M}}\left[a_{k,p}\delta_{j',k}\delta_{p,i}+a_{p,k}\delta_{j',p}\delta_{k,i}\right]\\
	&=\sum_{i'\in \mathcal{A}}\left[a_{k,p}c_{i',k}c^*_{M,p}+a_{p,k}c_{i',p}c^*_{M,k}\right]\ket{\psi_{i'}}\bra{\psi_{M}}.
\end{split}
\end{equation}
The brackets must vanish for an arbitrary choice of $a_{k,p}$ and $a_{p,k}$ which make $O_{k,p}$ a normalized observable. This requires that $c_{i',k}c^*_{M,p}=0$ and $c_{i',p}c^*_{M,k}=0$ for all $i',p,k\in \mathcal{A}$. Either $c_{i',k}=0$ for all $i',k\in \mathcal{A}$ or $c_{M,p}=0$ for all $p\in \mathcal{A}$. We now compute the third term:

\begin{equation}
\begin{split}
Q_3 O_{k,p} Q_1^\dagger
	&=\sum_{i'\in \mathcal{A}}c_{M,i'}\ket{\psi_{M}}\bra{\psi_{i'}}	\left[a_{k,p}\ket{\psi_k}\bra{\psi_p}+a_{p,k}\ket{\psi_p}\bra{\psi_k}\right]\sum_{i,j\in \mathcal{A}}c^*_{i,j}\ket{\psi_{j}}\bra{\psi_{i}}\\
	&=\sum_{i,j,i'\in \mathcal{A}}c_{M,i'}c^*_{i,j}\ket{\psi_{M}}\bra{\psi_{i}}	\left[a_{k,p}\delta_{i',k}\delta_{p,j}+a_{p,k}\delta_{i',p}\delta_{k,j}\right]\\
	&=\sum_{i\in \mathcal{A}}\left[a_{k,p}c_{M,k}c^*_{i,p}+a_{p,k}c_{M,p}c^*_{i,k}\right]\ket{\psi_{M}}\bra{\psi_{i}}.
\end{split}
\end{equation}
The conditions found above ensure that the brackets vanish. 

Given all of the constraints on the coefficients, there are two possible forms of $U$:
\begin{equation}
	\begin{split}
	U_1&=\sum_{i,j\in \mathcal{A}}c_{i,j}\ket{\psi_i}\bra{\psi_j}+\sum_{i\in \mathcal{B}}c_{i,l_i}\ket{\psi_i}\bra{\psi_{l_i}},\\
	U_2&=\sum_{i\in \mathcal{B}}c_{i,l_i}\ket{\psi_i}\bra{\psi_{l_i}}+\sum_{i\in \mathcal{A}}c_{i,L}\ket{\psi_i}\bra{\psi_{L}}+\sum_{i\in \mathcal{A}}c_{M,i}\ket{\psi_{M}}\bra{\psi_i}.
	\end{split}
\end{equation} 
We show that $U_2$ is not a unitary. We compute 
\begin{equation}
\begin{split}
	U_2^\dagger U_2
	&=\left[\sum_{i'\in \mathcal{B}}c^*_{i',l_{i'}}\ket{\psi_{l_{i'}}}\bra{\psi_{i'}}+\sum_{i'\in \mathcal{A}}c^*_{i',L}\ket{\psi_{L}}\bra{\psi_{i'}}+\sum_{i'\in \mathcal{A}}c^*_{M,i'}\ket{\psi_{i'}}\bra{\psi_{M}}\right]\\
	&\hspace{5mm}\times\left[\sum_{i\in \mathcal{B}}c_{i,l_i}\ket{\psi_i}\bra{\psi_{l_i}}+\sum_{i\in \mathcal{A}}c_{i,L}\ket{\psi_i}\bra{\psi_{L}}+\sum_{i\in \mathcal{A}}c_{M,i}\ket{\psi_{M}}\bra{\psi_i}\right]\\
	&=\sum_{i'\in \mathcal{B}}c^*_{i',l_{i'}}\ket{\psi_{l_{i'}}}\bra{\psi_{i'}}\left[\sum_{i\in \mathcal{B}}c_{i,l_i}\ket{\psi_i}\bra{\psi_{l_i}}+\sum_{i\in \mathcal{A}}c_{M,i}\ket{\psi_{M}}\bra{\psi_i}\right]+\sum_{i,i'\in \mathcal{A}}c^*_{i',L}c_{i,L}\ket{\psi_{L}}\bra{\psi_{i'}}\psi_i\rangle\bra{\psi_{L}}\\
	&\hspace{5mm}+\sum_{i'\in \mathcal{A}}c^*_{M,i'}\ket{\psi_{i'}}\bra{\psi_{M}}\left[\sum_{i\in \mathcal{B}}c_{i,l_i}\ket{\psi_i}\bra{\psi_{l_i}}+\sum_{i\in \mathcal{A}}c_{M,i}\ket{\psi_{M}}\bra{\psi_i}\right]\\
	&=\sum_{i,i'\in \mathcal{B}}c^*_{i',l_{i'}}c_{i,l_i}\ket{\psi_{l_{i'}}}\bra{\psi_{l_i}}\delta_{i,i'}+\sum_{\substack{i\in \mathcal{A},\\i'\in \mathcal{B}}}c^*_{i',l_{i'}}c_{M,i}\ket{\psi_{l_{i'}}}\bra{\psi_i}\delta_{i',M}+\sum_{i,i'\in \mathcal{A}}c^*_{i',L}c_{i,L}\ket{\psi_{L}}\bra{\psi_{L}}\delta_{i,i'}\\
	&\hspace{5mm}+\sum_{\substack{i'\in \mathcal{A},\\i\in \mathcal{B}}}c^*_{M,i'}c_{i,l_i}\ket{\psi_{i'}}\bra{\psi_{l_i}}\delta_{i,M}+\sum_{i,i'\in \mathcal{A}}c^*_{M,i'}c_{M,i}\ket{\psi_{i'}}\bra{\psi_i}\\
	&=\sum_{i\in \mathcal{B}}c^*_{i,l_{i}}c_{i,l_i}\ket{\psi_{l_{i}}}\bra{\psi_{l_i}}+\sum_{\substack{i\in A}}c^*_{M,l_{M}}c_{M,i}\ket{\psi_{l_{M}}}\bra{\psi_i}+\sum_{i\in A}\abs{c_{i,J}}^2\ket{\psi_{L}}\bra{\psi_{L}}\\
	&\hspace{5mm}+\sum_{\substack{i'\in \mathcal{A}}}c^*_{M,i'}c_{M,l_M}\ket{\psi_{i'}}\bra{\psi_{l_M}}+\sum_{i,i'\in \mathcal{A}}c^*_{M,i'}c_{M,i}\ket{\psi_{i'}}\bra{\psi_i}.
\end{split}
\end{equation}
Unitarity requires that $U_2^\dagger U_2=I=\sum_{i\in \mathcal{A}\cup \mathcal{B}}\ket{\psi_i}\bra{\psi_i}$. Therefore, we require that
\begin{equation}
	c^*_{M,i'}c_{M,i}=\delta_{i,i'}
\end{equation}
for all $i,i'\in \mathcal{A}$. However, no such $c_{M,i}$ exists, so $U_2$ is not unitary and $U_2\notin \mathcal{U}_{\mathcal{O}_F}$.

We now find the conditions that make $U_1$ unitary:
\begin{equation}
\begin{split}
U_1^\dagger U_1
	&=\left[\sum_{i',j'\in \mathcal{A}}c^*_{i',j'}\ket{\psi_{j'}}\bra{\psi_{i'}}+\sum_{i'\in \mathcal{B}}c^*_{i',l_{i'}}\ket{\psi_{l_{i'}}}\bra{\psi_{i'}}\right]\left[\sum_{i,j\in \mathcal{A}}c_{i,j}\ket{\psi_i}\bra{\psi_j}+\sum_{i\in \mathcal{B}}c_{i,l_i}\ket{\psi_i}\bra{\psi_{l_i}}\right]\\
	&=\sum_{i,j,i',j'\in \mathcal{A}}c^*_{i',j'}c_{i,j}\ket{\psi_{j'}}\bra{\psi_{i'}} \psi_i\rangle\bra{\psi_j}+\sum_{i,i'\in \mathcal{B}}c^*_{i',l_{i'}}c_{i,l_i}\ket{\psi_{l_{i'}}}\bra{\psi_{i'}}\psi_i\rangle\bra{\psi_{l_i}}\\
	&=\sum_{i,j,i',j'\in \mathcal{A}}c^*_{i',j'}c_{i,j}\ket{\psi_{j'}}\bra{\psi_j}\delta_{i,i'}+\sum_{i,i'\in \mathcal{B}}c^*_{i',l_{i'}}c_{i,l_i}\ket{\psi_{l_{i'}}}\bra{\psi_{l_i}}\delta_{i,i'}\\
	&=\sum_{i,j,j'\in \mathcal{A}}c^*_{i,j'}c_{i,j}\ket{\psi_{j'}}\bra{\psi_j}+\sum_{i\in \mathcal{B}}\abs{c_{i,l_i}}^2\ket{\psi_{l_{i}}}\bra{\psi_{l_i}}.
\end{split}
\end{equation}
Unitarity requires that 
 \begin{equation}\label{Eq:li2}
 	\{l_i\}_{i\in\mathcal{B}}=\mathcal{B},
 \end{equation}
 \begin{equation}\label{Eq:cli2}
 	\abs{c_{i,l_i}}^2=1, \ \forall i\in \mathcal{B},
 \end{equation}
 \begin{equation}
 	\sum_{i\in \mathcal{A}}c^*_{i,j'}c_{i,j}=\delta_{j',j}, \quad j,j'\in \mathcal{A}.
 \end{equation}

 We ensure that the conditions previously found for $c_{i,l_i}$ are consistent with Eqs.~\eqref{Eq:li2} and \eqref{Eq:cli2}. For $k\in \mathcal{B}$, recall that there exists an index $m_k\in \mathcal{B}$ such that $c_{i,k}=0$ for $i\in \mathcal{B}$ satisfying $i\neq m_k$. Also, there exists an index $l_k\in \mathcal{B}$ such that $c_{k,j}=0$ for $j\in \mathcal{B}$ and $j\neq l_k$. Assume that, for $k_1\in \mathcal{B}$, $l_{k_1}=l'$. There exists an index $m_{l'}\in \mathcal{B}$ such that $c_{i,l'}=0$ for $i\in \mathcal{B}$ and $i\neq m_{l'}$. If $k_1\neq m_{l'}$, then $c_{k_1,l_{k_1}}=0$, contradicting Eq.~\eqref{Eq:cli2}. Therefore, $m_{l'}= k_1$. In other words, if $k_1,k_2\in\mathcal{B}$ and $l_{k_1}=l_{k_2}$, then $k_1=k_2$. This ensures that $\{l_k\}_{k\in \mathcal{B}}=\mathcal{B}$, consistent with Eq.~\eqref{Eq:li2}.

Lastly, for an arbitrary free observable $O=\sum_{k,p\in \mathcal{A}}a_{k,p}\ket{\psi_k}\bra{\psi_p}+\sum_{k\in \mathcal{B}}a_{k,k}\ket{\psi_k}\bra{\psi_k}$, we verify that $U_1^\dagger O U_1\in \mathcal{O}_F$:
\begin{equation}
\begin{split}
	U_1^\dagger O U_1
	&=\left[\sum_{i',j'\in \mathcal{A}}c^*_{i',j'}\ket{\psi_{j'}}\bra{\psi_{i'}}+\sum_{i'\in \mathcal{B}}c^*_{i',l_{i'}}\ket{\psi_{l_{i'}}}\bra{\psi_{i'}}\right]\left[\sum_{k,p\in \mathcal{A}}a_{k,p}\ket{\psi_k}\bra{\psi_p}+\sum_{k\in \mathcal{B}}a_{k,k}\ket{\psi_k}\bra{\psi_k}\right]\\
	&\hspace{5mm}\times\left[\sum_{i,j\in \mathcal{A}}c_{i,j}\ket{\psi_i}\bra{\psi_j}+\sum_{i\in \mathcal{B}}c_{i,l_i}\ket{\psi_i}\bra{\psi_{l_i}}\right]\\
	&=\sum_{i,j,i',j',k,p\in \mathcal{A}}c^*_{i',j'}c_{i,j}a_{k,p}\ket{\psi_{j'}}\bra{\psi_{j}}\delta_{i',k}\delta_{p,i}+\sum_{i,i',k\in \mathcal{B}}c^*_{i',l_{i'}}c_{i,l_i}a_{k,k}\ket{\psi_{l_{i'}}}\bra{\psi_{l_i}}\delta_{i',k}\delta_{k,i}\\
	&=\sum_{j,j',k,p\in \mathcal{A}}c^*_{k,j'}c_{p,j}a_{k,p}\ket{\psi_{j'}}\bra{\psi_{j}}+\sum_{k\in \mathcal{B}}c^*_{k,l_{k}}c_{k,l_k}a_{k,k}\ket{\psi_{l_{k}}}\bra{\psi_{l_k}}\\
	&=\sum_{j,j'\in \mathcal{A}}\left[\sum_{k,p\in \mathcal{A}}c^*_{k,j'}c_{p,j}a_{k,p}\right]\ket{\psi_{j'}}\bra{\psi_{j}}+\sum_{k\in \mathcal{B}}a_{k,k}\ket{\psi_{l_{k}}}\bra{\psi_{l_k}}\\
	&\in \mathcal{O}_F.
\end{split}
\end{equation}
 
Hence, all unitaries in $\mathcal{U}_{\mathcal{O}_F}$ have the form of $U_1$, implying that $\mathcal{U}_{\mathcal{O}_F}=\mathcal{U}_F$.

\section{Properties of the observable non-revival monotone}
We prove faithfulness. For any $O\in\mathcal{O}$, $G(O)= \abs{\langle O(2\pi T),O\rangle}^2\leq \norm{O(2\pi T)}_2^2\norm{O}_2^2=1$ and $\mathcal{G}(O)\geq 0$. For any $U_F\in \mathcal{U}_F$ and $O\in \mathcal{O}_F$, $U_F^\dagger O U_F\in\mathcal{O}_F$, so $G(U_F^\dagger O U_F)=1$ and $\mathcal{G}(O)=0$. Now take $O\notin \mathcal{O}_F$. From the below lemma, it follows that $U_F^\dagger O U_F\notin \mathcal{O}_F$ and $G(U^\dagger_F O U_F)<1$, so $\mathcal{G}(O)>0$. Hence, $\mathcal{G}$ is faithful.
\begin{lemma}
For $U_F\in \mathcal{U}_F$ and $O\notin \mathcal{O}_F$, then $U_F^\dagger O U_F \notin \mathcal{O}_F$. 
\end{lemma}

\begin{proof}Assume that $U_F^\dagger  O U_F=O_2 \in \mathcal {O}_F$. Then this implies that $U_F O_2 U^\dagger_F \in \mathcal{O}_F$, since $U_F^\dagger \in \mathcal{U}_F$. This yields a contradiction. Therefore, $U_F^\dagger O U_F \notin \mathcal{O}_F$.
\end{proof}

One can see that $\mathcal{G}(U_F^\dagger OU_F)=\max_{U_F'\in \mathcal{U}_F}\left\{
1-G(U_F'^\dagger U_F^\dagger O U_F U_F')\right\}=\max_{U_F'\in \mathcal{U}_F}\left\{
1-G(U_F'^\dagger  O  U_F')\right\}=\mathcal{G}( O)$. This proves invariance.

\section{Proof of Theorem~\ref{Theorem:OTOCBound}}

Let $O_1$ and $O_2$ be Pauli operators. One can write
\begin{equation}
	O_1(t)=c_{O_2}(t)O_2+\sum_{P\in \mathcal{P}_2^{\otimes n},P\neq O_2}c_{P}(t)P.
\end{equation}
Each coefficient $c_P(t)$ is real for all $P\in \mathcal{P}_2^{\otimes n}$, since 
\begin{equation}
	c_P(t)=\langle O_1(t)P\rangle=\langle P O_1 (t)\rangle=\langle P^\dagger O_1^\dagger (t)\rangle=\langle (O_1 (t)P)^\dagger\rangle=\langle O_1 (t)P\rangle^*=c_P^*(t).
\end{equation}

Then,
\begin{equation}
\begin{split}
	\langle O_1(t)O_2 O_1(t)O_2\rangle
	&= \frac{1}{d}\tr{\left(c_{O_2}(t)O_2+\sum_{P\in \mathcal{P}_2^{\otimes n},P\neq O_2}c_{P}(t)P\right)O_2\left(c_{O_2}(t)O_2+\sum_{P'\in \mathcal{P}_2^{\otimes n},P'\neq O_2}c_{P}(t)P'\right)O_2}\\
	&= \frac{1}{d}\tr{\left(c_{O_2}(t)I+\sum_{P\in \mathcal{P}_2^{\otimes n},P\neq O_2}c_{P}(t)PO_2\right)\left(c_{O_2}(t)I+\sum_{P'\in \mathcal{P}_2^{\otimes n},P'\neq O_2}c_{P}(t)P'O_2\right)}\\
	&= \frac{1}{d}\mathrm{Tr}\Bigg\{c^2_{O_2}(t)I+\sum_{P'\in \mathcal{P}_2^{\otimes n},P'\neq O_2}c_{O_2}(t)c_{P'}(t)P'O_2+\sum_{P\in \mathcal{P}_2^{\otimes n},P\neq O_2}c_{P}(t)c_{O_2}(t)PO_2\\
	&\hspace{10mm}+\sum_{\substack{P\in \mathcal{P}_2^{\otimes n},P\neq O_2,\\P'\in \mathcal{P}_2^{\otimes n},P'\neq O_2}}c_{P}(t)c_{P'}(t)P'O_2PO_2\Bigg\}\\
	&= c^2_{O_2}(t)+\frac{1}{d}\sum_{\substack{P\in \mathcal{P}_2^{\otimes n},P\neq O_2,\\P'\in \mathcal{P}_2^{\otimes n},P'\neq O_2}}c_{P}(t)c_{P'}(t)f(P,O_2)\tr{P'PO_2^2}\\
	&=c^2_{O_2}(t)+\sum_{P\in \mathcal{P}_2^{\otimes n},P\neq O_2}c_{P}^2(t)f(P,O_2)\\
	&\geq c^2_{O_2}(t)-\sum_{P\in \mathcal{P}_2^{\otimes n},P\neq O_2}c_{P}^2(t)\\
	&= c^2_{O_2}(t)-(1-c^2_{O_2}(t))\\
	&= 2\langle O_1(t),O_2\rangle^2-1.
\end{split}
\end{equation}
Above, $f(P,O_2)=1$ if $P$ commutes with $O_2$ and $f(P,O_2)=-1$ if $P$ anti-commutes with $O_2$.

It follows that for $O_1=O_2$ and $t=2\pi T$,
\begin{equation}
\begin{split}
	\langle O_1(2\pi T)O_1 O_1(2\pi T)O_1\rangle
	&\geq 2\langle O_1(2\pi T),O_1\rangle^2-1\\
	&=  2G(O_1)-1\\
	&\geq  2\min_{U_F\in \mathcal{U}_F}G(U^\dagger_FO_1 U_F)-1\\
	&=  2(1-\mathcal{G}(O_1))-1\\
	&=  1-2\mathcal{G}(O_1).
\end{split}
\end{equation}
Also, by the Cauchy-Schwarz inequality,
\begin{equation}
	\langle O_1(2\pi T)O_1 O_1(2\pi T)O_1\rangle \leq \abs{\langle O_1(2\pi T)O_1 O_1(2\pi T),O_1\rangle}\leq  \norm{O_1(2\pi T)O_1 O_1(2\pi T)}_2\norm{O_1}_2=1.
\end{equation}

\section{Computation for Hayden-Preskill Decoding Protocol}
We compute the decoding fidelity in the $t=0$ case, where $U(0)=I$. First, we compute the average OTOC,
\begin{equation}
\begin{split}
\av{P_A, P_D}\OTOC(P_A,P_D;I)
	&=\av{P_A, P_D}\langle P_A P_DP_AP_D\rangle\\
	&=\frac{1}{d_A^2d_D^2d}\sum_{P_{A\backslash D},P_{A\cap D}}\sum_{P'_{D\backslash A},P'_{D\cap A}}\tr{ P_{A\backslash D}P_{A\cap D}P'_{D\backslash A}P'_{D\cap A}P_{A\backslash D}P_{A\cap D}P'_{D\backslash A}P'_{D\cap A}}\\
	&=\frac{1}{d_A^2d_D^2d}\sum_{P_{A\backslash D},P_{A\cap D}}\sum_{P'_{D\backslash A},P'_{D\cap A}}\tr{ P_{A\cap D}P'_{D\cap A}P_{A\cap D}P'_{D\cap A}}\\
	&=\frac{1}{d_A^2d_D^2d}\sum_{P_{A\backslash D},P_{A\cap D}}\sum_{P'_{D\backslash A}}d^2_{A\cap D}\av{P'_{D\cap A}}\tr{ P_{A\cap D}P'_{D\cap A}P_{A\cap D}P'_{D\cap A}}\\
	&=\frac{1}{d_A^2d_D^2d}\sum_{P_{A\backslash D},P_{A\cap D}}\sum_{P'_{D\backslash A}}d_{A\cap D}\tr{ P_{A\cap D}}\ptr{A\cap D}{P_{A\cap D}}\\
	&=\frac{1}{d_A^2d_D^2d}\sum_{P_{A\backslash D},P_{A\cap D}}\sum_{P'_{D\backslash A}}dd^2_{A\cap D}\delta_{P_{A\cap D},I_{A\cap D}}\\
    &=\frac{1}{d_A^2d_D^2d}\sum_{P_{A\backslash D},P'_{D\backslash A}}dd^2_{A\cap D}\\
	&=\frac{1}{d_A^2d_D^2d}dd^2_{A\cap D}d^2_{A\backslash D}d^2_{D\backslash A}\\
	&=\frac{1}{d_A^2d_D^2}d_D^2d^2_{A\backslash D}
	=\frac{d^2_{A\backslash D}}{d_A^2}.
\end{split}
\end{equation}
Therefore $F(U(0))=\frac{1}{d_{A\backslash D}^2 }$.

\section{Computation of weak measurement channel}

We describe the weak measurement formalism. We use $A$ to denote Alice's $d$-dimensional system with computational basis states $\{\ket{i}_A\}_{i=1}^d$ and $E$ to denote Eve's $d+1$-dimensional system with basis states $\{\ket{i}_E\}_{i=0}^d$. Define the controlled rotation as 
\begin{equation}
R_{\alpha}=\sum_{i=1}^{d}\ket{i}\bra{i}_A\otimes e^{-i\alpha X_{(E),i}}  
\end{equation}
where ${X_{(E),i}=\ket{i}\bra{0}_E+\ket{0}\bra{i}_E}$ and $\alpha\in (0,\pi/2)$. Define the dephasing channel on the ancilla system as ${\mathcal{N}_{(E)}[\cdot]=\sum_{i=0}^{d}\ket{i}\bra{i}_E[\cdot]\ket{i}\bra{i}_E}$.
The weak measurement channel of strength $p=\sin^2(\alpha)$ on $\rho\otimes \ket{0}\bra{0}_E$, where $\rho$ is a target state and $\ket{0}\bra{0}_E$ is an ancillary state, is
\begin{equation}
\begin{split}
	\mathcal{M}_p(\rho)
	&=\mathcal{N}_{(E)}[R_{\alpha}(\rho\otimes \ket{0}\bra{0}_E)R_{\alpha}^\dagger]
 \end{split}
 \end{equation}
We show how one can write this channel as
\begin{equation}
\mathcal{M}_p(\rho)=(1-p)\rho\otimes \ket{0}\bra{0}_E+p\sum_{i=1}^{d}(\ket{i}\bra{i}_A\rho \ket{i}\bra{i}_A)\otimes \ket{i}\bra{i}_E.
\end{equation}

We begin by rewriting the channel as
\begin{equation}
\begin{split}
	\mathcal{M}_p(\rho)&=\mathcal{N}_{(E)}[R_{\alpha}(\rho\otimes \ket{0}\bra{0}_E)R_{\alpha}^\dagger]\\
	&=\sum_{i=0}^{d}I_A\otimes\ket{i}\bra{i}_E\left[\sum_{j=1}^{d}(\ket{j}\bra{j}_A\otimes e^{-i\alpha X_{(E),j}})(\rho\otimes \ket{0}\bra{0}_E)\sum_{j'=1}^{d}(\ket{j'}\bra{j'}_A\otimes e^{i\alpha X_{(E),j'}})\right]I_A\otimes\ket{i}\bra{i}_E\\   &=\sum_{i=0}^{d}\sum_{j,j'=1}^{d}\ket{j}\bra{j}_A\rho \ket{j'}\bra{j'}_A\otimes  \ket{i}\bra{i}_E e^{-i\alpha X_{(E),j}}\ket{0}\bra{0}_Ee^{i\alpha X_{(E),j'}}\ket{i}\bra{i}_E.
\end{split}
\end{equation}

Using identities computed below, we can rewrite this as
\begin{equation}
\begin{split}
	\mathcal{M}_p(\rho)
	&=\sum_{i=0}^{d}\sum_{j,j'=1}^{d}\ket{j}\bra{j}_A\rho \ket{j'}\bra{j'}_A\otimes  \ket{i}\bra{i}_E (\delta_{i,0}\cos(\alpha)-i\delta_{i,j}\sin(\alpha))(\delta_{i,0}\cos(\alpha)+i\delta_{i,j'}\sin(\alpha))\\
    &=\sum_{i=0}^{d}\sum_{j,j'=1}^{d}\ket{j}\bra{j}_A\rho \ket{j'}\bra{j'}_A\otimes  \ket{i}\bra{i}_E \\
	&\hspace{10mm}\cdot(\delta_{i,0}\delta_{i,0}\cos^2(\alpha)+i\delta_{i,0}\delta_{i,j'}\sin(\alpha)\cos(\alpha)-i\delta_{i,j}\delta_{i,0}\sin(\alpha)\cos(\alpha)+\delta_{i,j}\delta_{i,j'}\sin^2(\alpha))\\
	&=\sum_{j,j'=1}^{d}\ket{j}\bra{j}_A\rho \ket{j'}\bra{j'}_A\otimes  \\
	&\hspace{10mm}(\delta_{0,0}\cos^2(\alpha)\ket{0}\bra{0}_E+i\delta_{0,j'}\sin(\alpha)\cos(\alpha)\ket{0}\bra{0}_E-i\delta_{0,j}\sin(\alpha)\cos(\alpha)\ket{0}\bra{0}_E+\delta_{j,j'}\sin^2(\alpha)\ket{j}\bra{j}_E)\\
	&=\sum_{j,j'=1}^{d}\ket{j}\bra{j}_A\rho \ket{j'}\bra{j'}_A\otimes  (\cos^2(\alpha)\ket{0}\bra{0}_E+\delta_{j,j'}\sin^2(\alpha)\ket{j}\bra{j}_E)\\
	&=(1-p) \rho \otimes  \ket{0}\bra{0}_E+p\sum_{j=1}^{d}\ket{j}\bra{j}_A\rho \ket{j}\bra{j}_A\otimes \ket{j}\bra{j}_E.
\end{split}
\end{equation}
This is the form of the channel presented in the main text. In the above, we used that 
\begin{equation}
\begin{split}
	\bra{i}e^{-i\alpha X_{(a),j}}\ket{0}_E
	&=\bra{i}\left[\sum_{k=0}^{\infty}\tfrac{1}{k!}(-i\alpha)^kX_{(E),j}^k\right]\ket{0}_E\\
	&=\sum_{k=\mathrm{even}}\tfrac{1}{k!}(-i\alpha)^k \bra{i} X_{(E),j}^k\ket{0}_E+\sum_{k=\mathrm{odd}}\tfrac{1}{k!}(-i\alpha)^k \bra{i} X_{(E),j}^k\ket{0}_E\\
	&=\sum_{k=\mathrm{even}}\tfrac{1}{k!}(-i\alpha)^k \bra{i} (\ket{0}\bra{0}_E+\ket{j}\bra{j}_E)\ket{0}_E+\sum_{k=\mathrm{odd}}\tfrac{1}{k!}(-i\alpha)^k \bra{i} (\ket{0}\bra{j}_E+\ket{j}\bra{0}_E)\ket{0}_E\\
	&=\delta_{i,0}\sum_{k=\mathrm{even}}\tfrac{1}{k!}(-i\alpha)^k +\delta_{i,j}\sum_{k=\mathrm{odd}}\tfrac{1}{k!}(-i\alpha)^k \\
	&=\delta_{i,0}\sum_{k=0}^{\infty}\tfrac{1}{(2k)!}(-i)^{2k}\alpha^{2k} +\delta_{i,j}\sum_{k=0}^{\infty}\tfrac{1}{(2k+1)!}(-i\alpha)^{2k+1} \\
	&=\delta_{i,0}\sum_{k=0}^{\infty}\tfrac{1}{(2k)!}(-1)^{k}\alpha^{2k} -i\delta_{i,j}\sum_{k=0}^{\infty}\tfrac{1}{(2k+1)!}(-1)^{k}(\alpha)^{2k+1}\\
	&=\delta_{i,0}\cos(\alpha)-i\delta_{i,j}\sin(\alpha).
\end{split}
\end{equation}
In line three, we used that, when $k$ is even,
\begin{equation}
\begin{split}
	X_{(E),j}^k
	&=X_{(E),j}^{k-2}(\ket{j}\bra{0}_E+\ket{0}\bra{j}_E)(\ket{j}\bra{0}_E+\ket{0}\bra{j}_E)\\
	&=X_{(E),j}^{k-2}(\ket{j}\bra{j}_E+\ket{0}\bra{0}_E)\\
	&=(\ket{j}\bra{j}_E+\ket{0}\bra{0}_E)^{k/2}\\
	&=\ket{j}\bra{j}_E+\ket{0}\bra{0}_E.
\end{split}
\end{equation}
When $k$ is odd,
\begin{equation}
\begin{split}
	X_{(E),j}^k
	&=X_{(E),j}X_{(E),j}^{k-1}\\
	&=(\ket{j}\bra{0}_E+\ket{0}\bra{j}_E)(\ket{j}\bra{j}_E+\ket{0}\bra{0}_E)\\
	&=\ket{j}\bra{0}_E+\ket{0}\bra{j}_E.
\end{split}
\end{equation}
It follows that 
\begin{equation}
\begin{split}
	\bra{0}e^{i\alpha X_{(E),j}}\ket{i}_E
	&=\delta_{i,0}\cos(\alpha)+i\delta_{i,j}\sin(\alpha).
\end{split}
\end{equation}

\end{appendix}

\end{document}